\documentclass{article}
\usepackage[normalem]{ulem}
\usepackage{amssymb,amsmath,amsfonts,bbm,eurosym,geometry,graphicx,caption,color,setspace,comment,footmisc,natbib,pdflscape,array,enumerate,nomencl,etoolbox,mathtools,booktabs,makecell,multirow,multicol,float}
\usepackage{subfigure}
\usepackage{booktabs}
\usepackage{authblk}
\usepackage{tikz}
\usetikzlibrary{decorations.pathreplacing}
\usepackage[ruled]{algorithm2e}
\usepackage[toc,page]{appendix}
\usepackage{accents}
\usepackage[thmmarks,amsmath,thref]{ntheorem}
\newtheorem{theorem}{Theorem}[section]
\newtheorem{definition}{Definition}[section]
\newtheorem{lemma}{Lemma}[section]

\newtheorem{assumption}{Assumption}[section]

\newtheorem{proposition}{Proposition}[section]

\newtheorem{remark}{Remark}[section]
\newtheorem{open problem}{Open Problem}
\theorembodyfont{\normalfont}

\theoremsymbol{{$\square$}}
\newtheorem*{proof}{Proof}

\newcommand{\bp}{b^{\prime}}

\numberwithin{equation}{section}

\usepackage[colorlinks=true,linkcolor=blue]{hyperref}
\hypersetup{%
  citecolor=blue
}

\begin{document}

\title{Optimal Random Taxation\thanks{
We thank Xiaodong Zhu for motivating us to examine the complete markets problem.
Shen thanks the support by the National Key R\&D Program of China with project number 2021YFA1001200, and the NSFC with grant number 12171013. Kubler thanks the SNF for financial support.}}
\author{Chengfeng Shen\footnote{Peking University. Email: \href{mailto:shencf1999@pku.edu.cn}{shencf1999@pku.edu.cn}} \and Felix K\"{u}bler\footnote{University of Zurich. Email: \href{mailto:fkubler@gmail.com}{fkubler@gmail.com}}}

		\newcommand\tbbint{{-\mkern -16mu\int}}
	\newcommand\tbint{{\mathchar '26\mkern -14mu\int}}
	\newcommand\dbbint{{-\mkern -19mu\int}}
	\newcommand\dbint{{\mathchar '26\mkern -18mu\int}}
	\newcommand\bint{
		{\mathchoice{\dbint}{\tbint}{\tbint}{\tbint}}
	}
	\newcommand\bbint{
		{\mathchoice{\dbbint}{\tbbint}{\tbbint}{\tbbint}}
	}
	\maketitle
\begin{abstract}
We consider the Ramsey problem of optimal government policy and give examples where lotteries are essential for the optimal solution. 
When markets are complete and the representative agent's coefficient of relative prudence is bounded, lotteries are typically not optimal.
For the incomplete markets case,
we prove analytically that lotteries improve welfare for some calibrations with CRRA utility, and we give a simple numerical example for this claim. 
When the agent's coefficient of relative risk aversion is unbounded for small consumption, lotteries typically improve welfare and can be used to approximate the first-best solution arbitrarily well in both the complete and the incomplete markets cases. However, in the complete markets case,
when debt constraints are not too tight and the representative agent is sufficiently patient, lotteries can be purified in the sense that there is a deterministic solution that gives the same value.
\end{abstract}

\noindent \textit{JEL classification}: \\
%

\clearpage

\section{Introduction}
 Following \cite{barro1979determination},  we consider the  classical question of how a government should optimally finance an exogenous stream of spending via distortionary taxes and debt. We consider the models from \cite{lucas1983optimal} and from \cite{aiyagari2002optimal}\footnote{
\cite{lucas1983optimal} were the first to consider the problem in  a full general-equilibrium ``Ramsey framework'': a benevolent government chooses labor tax paths to finance spending while respecting households' optimizing behavior. Crucially, they assume the government can issue a complete set of state-contingent claims.
\cite{aiyagari2002optimal} relax the complete-markets assumption, restricting the government to a single non-contingent bond.} but allow for randomized tax policies. In these models, the constraints that capture the optimal intertemporal behavior of households need not be concave, implying that the government faces a non-convex optimization problem.
 Non-convexities imply that deterministic optima may be dominated by randomized solutions\footnote{Throughout the paper, we use the terms ``lottery solution'' and ``randomized solution'' interchangeably to refer to the optimal, possibly random, solution of a nonlinear optimization problem.} (see, e.g., \cite{myerson1982optimal,prescott1984pareto,stiglitz1982self}).
In general, in the presence of non-convexities, the scope for improvement through randomization depends on the economic model, in particular on which constraints can be assumed to hold only in expectation.
\cite{shen2025recursive} show that
in an intertemporal setting with forward-looking constraints, randomization always relaxes those constraints even if within a given period all agents can make choices contingent on the outcome of the lottery. In this paper, we apply the theoretical results in \cite{shen2025recursive} to the Ramsey problem.

\cite{citanna2024taxspots} were  the first to establish that for many preference specifications, lotteries improve the optimal solution to the planning problems in \cite{lucas1983optimal} and in \cite{aiyagari2002optimal}. 
 They argue that if households are sufficiently prudent, then
adding uncertainty to disposable income increases savings, and risk-averse households
are more than compensated for the additional consumption uncertainty by an increased expected consumption. Although they raise the theoretical possibility of optimal lotteries in these models, it remains a quantitative question whether lotteries improve welfare in standard calibrations of the models and whether this improvement is economically significant.
In this paper, we tackle this question and show that optimal lotteries emerge naturally when markets are incomplete (i.e., in the version of the model from \cite{aiyagari2002optimal}), whereas in the complete markets version of the model they exist only for ``unusual'' preferences.

We first examine the model with complete markets (\citealp{lucas1983optimal}) and derive a simple condition on preferences that ensures that lotteries are never optimal. Constant relative risk aversion (CRRA) utility over consumption satisfies this condition.
 When the coefficient of relative risk aversion is unbounded for small consumption (the simplest example being shifted log utility, $u(c)=\log(c-\zeta)$, $\zeta>0$),  the government can, in the limit,
 finance its expenditures without ever taxing labor by introducing lotteries where with arbitrarily low probability consumption is pushed very close to its lower bound. Agents will be willing to buy government debt at any price to insure against this low-probability outcome. Although lotteries are optimal for this utility specification, there is another, deterministic solution that gives the same welfare when patience $ \beta $ is sufficiently large and the government does not face explicit debt constraints. The lottery solution can be ``purified'' (see, e.g., \citealp{fudenberg1991dispensability}, for a similar logic in the context of infinitely repeated games).

We then move to the model with incomplete markets and present an example with CRRA utility where lotteries improve welfare if initial debt is sufficiently high. We first illustrate the intuition in a two-period model: when entering the period with high debt, the government optimally decreases today's borrowing rate by introducing a low-probability \emph{bad state} with high labor taxes in the next period. This increases precautionary savings of the household and hence allows for lower taxes  in the first period. The simple two-period example is extended to a full-blown infinite-horizon example where we can demonstrate the importance of lotteries with CRRA utility.
As in the complete markets case, when utility over consumption exhibits unbounded relative risk aversion for small consumption and the lower bound on consumption can be arbitrarily small, lotteries are optimal, and the optimal lottery solution converges to the first-best.

Finally, we conduct a more quantitative analysis of lotteries in the incomplete markets model.
Building on the theoretical foundations developed in \cite{shen2025recursive}, we develop a simple algorithm to solve the functional equation that characterizes optimal lotteries and to construct an optimal lottery solution.
We then consider a large range of examples and show that in many of them optimal lotteries occur, while in many others, lotteries never improve welfare because they only occur in regions of the endogenous state space that are never visited in the optimal solution.
Although most of the analysis is conducted for CRRA utility over consumption, 
we point out that shifted log utility can lead to very large lotteries even if consumption is bounded away from its lower bound (i.e., on its domain the utility function exhibits bounded relative risk aversion).

The remainder of the paper is organized as follows. In Section~\ref{sec:model} we introduce the formal model and define lottery solutions. In Section~\ref{sec:cm} we analyze the model with complete markets, give sufficient conditions that rule out lotteries, show that with unbounded relative risk aversion lotteries approximate the first best, and discuss when optimal lotteries can be purified. In Section~\ref{sec:im} we provide a two-period and an analytical infinite-horizon example of optimal lotteries in a model with incomplete markets and extend the first-best result to this setting. 
In Section~\ref{sec:num} we consider a variety of numerical examples and show that the existence of optimal lotteries in this model depends on a variety of factors.
Section~\ref{sec:concl} concludes. In the Appendix, we collect the proofs, additional results and describe the computational method in detail.

\section{The Ramsey problem}\label{sec:model}
We consider a discrete-time, infinite-horizon economy with a government that maximizes the welfare of a representative agent who derives utility from consumption and leisure. The government faces Markovian expenditure shocks $ (g_t) $, which it has to finance by issuing debt or taxing labor. The government can announce a randomized tax policy where in each period taxes depend on the outcome of a public randomization device (lottery). After the lottery is conducted, the agents decide how much to work and save in the current period.\footnote{In the first period, lotteries can never improve welfare, as the agent's decisions are made after the lottery is realized.} Crucially, the savings decision will be determined by expectations about future lotteries.

We assume that government expenditures \(g_t\) lie in a finite set $\mathcal{G}\subset[\underline{g},\bar{g}]$, where \(0\le\underline{g}<\bar{g}<1\). The feasible consumption of an agent is \(c\ge\zeta\) for some \(\zeta\ge0\), and the utility of consuming \(c\) is \(u(c)\). The agent's feasible labor supply is \(0\le\ell\le1\), and the utility from supplying labor \(\ell\) is \(v(\ell)\). 
 In addition, the government is subject to constraints\footnote{These bounds are imposed for mathematical convenience, in particular to ensure an interior solution for the agent, a natural interpretation being that the government's tax choices are constrained away from both $-\infty$ and 1, which implicitly induces these bounds.} that ensure the agent's consumption and labor remain within certain ranges, namely \(c\ge \underline{c}\) and \(\ell\le \bar{\ell}\), where \(\underline{c}>\zeta\) and \(\bar{g}+\underline{c}\le\bar{\ell}<1\).
Initial expenditure and initial debt are given and denoted by \(g_0\) and \(b_{-1}\), respectively.

We consider a version of the model where the government has access to a complete set of Arrow securities (\citealp{lucas1983optimal}) and a version with a single one-period bond (\cite{aiyagari2002optimal}). In both models, a non-negative lump-sum transfer is available in each period.

\subsubsection*{Complete Markets}
When markets are complete and the government and the agent do not face explicit borrowing constraints\footnote{For most of the paper, we assume that implicit borrowing constraints are of the form $ \sup_t |b_t| < \infty $. For the incomplete markets case we will also examine explicit borrowing constraints.}, the sequential budget constraints can be combined into a single constraint, and the deterministic planner's problem can be written as follows  (for a detailed derivation see Appendix \ref{app:cm_derivation}; for technical details on the lottery solution see \cite{shen2025recursive})
  \begin{equation}  \label{eq:Lucas-S-d}
    \begin{aligned}    &\sup_{(\underline{c}+g_t\le \ell_t\le \bar{\ell})_{t=0}^\infty} \mathbb{E} \sum_{t=0}^{\infty} \beta^t \left(u(\ell_t-g_t)+v(\ell_t)\right) \\\mbox{ s.t. } &
     b_{-1} u'(\ell_0-g_0) \le \mathbb{E}_0 \sum_{t=0}^\infty \beta ^t (u'(\ell_t-g_t) (\ell_t-g_t) + v'(\ell_t) \ell_t).
\end{aligned}
\end{equation}
Consumption is pinned down by \(c_t=\ell_t-g_t\) and hence does not appear in \eqref{eq:Lucas-S-d}. 
To derive the lottery problem,
we first define the set \(L\) of feasible labor supply sequences \((\ell_t)_{t\ge 1}\) as
\[
L=\{(\ell_t)_{t\ge1}:\underline{c}+g_t\le \ell_t\le \bar{\ell},\,\forall t\ge 1\},
\]
and let
$\mathcal{P}(L)$ denote the space of probability measures on $L$.
In this formulation, $\ell_t$ is assumed to be a function of the history of exogenous shocks,
$g^t\in\mathcal{G}^t$; whenever possible, we simplify notation by writing $ \ell_t $.

If we allow for lotteries, the problem becomes
  \begin{equation}  \label{eq:Lucas-S-l}  
    \begin{aligned}    \sup_{\ell_0,P\in\mathcal{P}(L)} &u(\ell_0-g_0)+v(\ell_0)+\mathbb{E}_{g_0}^{(\ell_t) \sim P}  \sum_{t=1}^{\infty} \beta^t \left(u(\ell_t-g_t)+v(\ell_t)\right) \\\mbox{ s.t. } &
     b_{-1} u'(\ell_0-g_0) \le u'(\ell_0-g_0)(\ell_0-g_0)+v'(\ell_0)\ell_0+\mathbb{E}_{g_0}^{(\ell_t)\sim P}   \sum_{t=1}^\infty \beta ^t (u'(\ell_t-g_t) (\ell_t-g_t) + v'(\ell_t)\ell_t). 
\end{aligned}
\end{equation}
For any labor $\ell$ and consumption $c$, we define
$$
F(c,\ell)=u(c)+v(\ell),\quad G(c,\ell)=u'(c)c+v'(\ell)\ell.
$$
We introduce the Lagrangian of problem \eqref{eq:Lucas-S-d} as
\begin{equation}\label{eq:reduced_Ramsey_Lag}
    \mathcal{L}((\ell_t)_{t=0}^{\infty};\lambda)=\mathbb{E}_{g_0}\sum_{t=0}^{\infty}\beta^t\left[F(\ell_t-g_t,\ell_t)+\lambda G(\ell_t-g_t,\ell_t)\right]-\lambda b_{-1}u'(c_0).
\end{equation}
The following sup-inf formula of the Lagrangian \(\mathcal{L}\), which we refer to as the primal problem, represents the optimal value of \eqref{eq:Lucas-S-d}:
\begin{equation}\label{eq:primal_CM}
V^{p}=\sup_{(\ell_t)_{t=0}^\infty}\inf_{\lambda\ge 0}\mathbb{E}_{g_0}\sum_{t=0}^{\infty}\beta^t\left[F(\ell_t-g_t,\ell_t)+\lambda G(\ell_t-g_t,\ell_t)\right]-\lambda b_{-1}u'(c_0).
\end{equation}
By interchanging \(\sup_{(\ell_t)_{t=1}^{\infty}}\) and \(\inf_{\lambda\ge 0}\), we obtain the sup-inf-sup formula of the Lagrangian \(\mathcal{L}\), which we refer to as the dual problem:
\begin{equation}\label{eq:expost_CM}
V^{d}=\sup_{\ell_0}\inf_{\lambda\ge 0}\sup_{(\ell_t)_{t=1}^\infty}\mathbb{E}_{g_0}\sum_{t=0}^{\infty}\beta^t\left[F(\ell_t-g_t,\ell_t)+\lambda G(\ell_t-g_t,\ell_t)\right]-\lambda b_{-1}u'(c_0).
\end{equation}
By the minimax theorem, $V^d\ge V^p$. In Appendix \ref{app:regu}, we impose some regularity assumptions to ensure that the suprema and infima in \eqref{eq:expost_CM} are attained, i.e., the \(\sup\) and \(\inf\) can be replaced by \(\max\) and \(\min\).  It follows from \cite{shen2025recursive} that
the optimal value of the dual problem \eqref{eq:expost_CM} is equal to the optimal value of the lottery problem \eqref{eq:Lucas-S-l}.

\subsubsection*{Bond-only economy}
 When markets are incomplete, we can write the deterministic problem as
  \begin{equation}  \label{eq:Ramsey-d}
    \begin{aligned}    &\sup_{(\underline{c}+g_t\le \ell_t\le \bar{\ell})_{t=0}^{\infty}} \mathbb{E}_{0} \sum_{t=0}^{\infty} \beta^t \left(u(\ell_t-g_t)+v(\ell_t)\right) \\\mbox{ s.t. } &\beta b_t \mathbb{E}_tu'(\ell_{t+1}-g_{t+1}) + u'(\ell_t-g_t)(\ell_t-g_t)+v'(\ell_t)\ell_t\ge b_{t-1}u'(\ell_t-g_t) \ \forall t\\
     & \sup_t |b_t| < \infty.
\end{aligned}
\end{equation}
To be consistent with the complete markets formulation, we assume that neither the government nor the agent faces explicit borrowing constraints. For the case of incomplete markets, this assumption is not necessarily realistic.
\cite{feng2026ramsey} analyze some of its implications.
As in \cite{aiyagari2002optimal}, we assume in our numerical examples that the admissible bond holdings of the government lie in a closed interval, $ b_t \in [\underline{b},\overline{b}]$, and satisfy the agent's natural borrowing constraint.

Similarly to above,
we first define the set \(L\) of feasible labor supply and bond sequences \((\ell_t,b_t)_{t\ge 0}\) as
\[
L=\{(\ell_t,b_t)_{t\ge0}:\underline{c}+g_t\le \ell_t \le \bar{\ell},
b_t \in  [\underline{b},\overline{b}], \  
\forall t\ge 0\},
\]
and let
$\mathcal{P}(L)$ denote the space of probability measures on $L$.\footnote{For notational convenience, we allow a lottery in the first period as well. However, since the constraint described later must be satisfied after the realization of the first-period lottery, the optimum will not involve a lottery in the first period.} The lottery problem is
  \begin{equation}  \label{eq:Ramsey-l}
    \begin{aligned}  \sup_{P\in\mathcal{P}(L)} & \mathbb{E}_{0}^{(\ell_t,b_t) \sim P}  \sum_{t=0}^{\infty} \beta^t \left(u(\ell_t-g_t)+v(\ell_t)\right) \\\mbox{ s.t. } &\beta b_t \mathbb{E}_{t}^{(\ell_t,b_t)\sim P}u'(\ell_{t+1}-g_{t+1}) + u'(\ell_t-g_t)(\ell_t-g_t)+v'(\ell_t)\ell_t\ge b_{t-1}u'(\ell_t-g_t)  \ \forall t.
\end{aligned}
\end{equation}
As above, $ \ell_t $ and $ b_t $ should be read as $ \ell(g^t) $ and $ b(g^t)$. In the lottery problem, it is assumed that 
each constraint is satisfied for any \(g^t\in\mathcal{G}^t\), and also for the realized labor supplies and bonds, \(\ell_0(g_0),\,b_0(g_0),\,\dots,\ell_t(g^t),\,b_t(g^t)\), with the expectation taken with respect to the conditional distribution \(P(\cdot\mid \ell_0(g_0),b_0(g_0),\dots,\ell_t(g^t),b_t(g^t))\).
The optimal value is given by the solution to the following functional equation
(see \cite{shen2025recursive}). 
The state is $(b,\mu,g)$. Define
\begin{equation}\label{eq:mafu}
F(\gamma,\ell;b,\mu,g)=
u(\ell-g)+v(\ell)+\mu\,b\,u'(\ell-g)+\gamma\big[u'(\ell-g)(\ell-g-b)+\ell\,v'(\ell)\big],
\end{equation}
and let
\begin{equation} \label{eq:fe}  W(b,\mu,g) = \sup_{l,b'} \inf_{\gamma \ge 0} F(\gamma,l;b,\mu,g)+\beta\,\mathbb{E}\big[W(b',\gamma,g')\mid g\big].
\end{equation}
$W(b,\mu,g)$ then represents the optimal value of 
  \begin{equation}  \label{eq:Ramsey-l-mu}
    \begin{aligned}  \sup_{P\in\mathcal{P}(L)} & \mathbb{E}_{0}^{(\ell_t,b_t) \sim P} \mu b_{-1}u'(\ell_0-g_0)+  \sum_{t=0}^{\infty} \beta^t \left(u(\ell_t-g_t)+v(\ell_t)\right) \\\mbox{ s.t. } &\beta b_t \mathbb{E}_{t}^{(\ell_t,b_t)\sim P}u'(\ell_{t+1}-g_{t+1}) + u'(\ell_t-g_t)(\ell_t-g_t)+v'(\ell_t)\ell_t\ge b_{t-1}u'(\ell_t-g_t)  \ \forall t.
\end{aligned}
\end{equation}

\section{Lotteries: Complete Markets}\label{sec:cm}
We first examine the model with complete financial markets and show that while there are examples in which the optimal solution employs randomization, this is not the case for typical parameter choices.
The basic intuition for why random taxation may improve welfare is that the government can artificially create a small positive probability of some high-tax state; precautionary saving motives then lower the interest rate and allow it to roll over its debt at lower cost (see also \cite{citanna2024taxspots}). However, we show that two distinct mechanisms prevent lotteries from improving welfare. First, when the relative prudence of the agent's utility function is not too high, the precautionary saving motive is not strong enough to generate a welfare improvement. In particular, with isoelastic utility, lotteries are never optimal, independently of the coefficient of relative risk aversion. 

If relative risk aversion is unbounded as $ c \rightarrow \zeta $, the opposite picture emerges: not only is it optimal for the government to randomize over future taxes, it can also approximate the first-best solution arbitrarily well if the lower bound on consumption becomes small, $ \underline{c} \rightarrow \zeta $.

However, when the discount factor is sufficiently large, an optimal lottery can always be purified in the sense that there is a deterministic sequence of taxes that achieves the same welfare. The intuition is that the government can commit to a high-tax state deterministically at some future date-event and thereby achieve the same effect as creating a positive probability of a high-tax outcome next period.

\subsection{A class of economies without optimal lotteries}
In this subsection, we introduce conditions under which a lottery does not improve welfare.
We make the following assumptions on utility.
\begin{assumption}\label{ass:v}
The utility functions for labor and consumption satisfy
\begin{enumerate}
\item\label{enum_assv_1}
    The utility function for labor, $v(\ell)$, satisfies $v'(\ell)<0,\,v''(\ell)<0,$ and $v'''(\ell)<0$. The utility function for consumption, $u(c)$, satisfies $u'(c)>0$ and $u''(c)<0$.
\item\label{enum_assv_2}     There exists $M>0$ such that the relative prudence of $u$ is bounded above by $M$, that is,
$$
-\frac{u'''(c)c}{u''(c)}\le M \quad \text{for all } c\in(0,1].
$$
    \end{enumerate}
    \end{assumption}
    
    Many utility functions commonly used to model labor supply satisfy part one of the condition (for instance, the CRRA type
$v(\ell)=\psi \frac{(1-\ell)^{1-\sigma}}{1-\sigma}$ $(\psi>0,\ \sigma>0)$
, or the constant Frisch elasticity type
$v(\ell)=-\frac{\ell^{1+1/\eta}}{1+1/\eta}$ with $0<\eta<\frac{1}{2}$).
Boundedness of relative prudence is satisfied for almost all utility functions used in practice.
When the bound $M$ is sufficiently small, \eqref{eq:Lucas-S-d} is a convex programming problem in $(\ell_t)_{t\ge 1}$ and has a unique solution that rules out lotteries. One can verify that its Lagrangian is concave for all $ \lambda \ge 0 $.
\begin{proposition}\label{prop:smallM}
    Under Assumption \ref{ass:v}  with $M\le 2$, there is no lottery improvement.
\end{proposition}

When $M>2$, the following lemma shows that, when $\lambda$ is small, the Lagrangian is still concave.
\begin{lemma}\label{lem:M>2}
    Under Assumption \ref{ass:v} with $M>2$, $F(c_t,\ell_t)+\lambda G(c_t,\ell_t)$ is concave when $0\le \lambda\le \frac{1}{M-2}$.
\end{lemma}

The following assumption\footnote{The intuition for this assumption is that, for a sufficiently small Lagrange multiplier such that the Lagrangian is concave, the constraint is already nonbinding.} can be easily checked and turns out to be sufficient for the optimal $ \lambda $ to be sufficiently small.
\begin{assumption}\label{ass:M>2}
Let $\ell_0^*$ be the optimal labor supply in the first period, i.e.,
$$
\ell_0^*\in\arg\max_{\ell_0\in L_0}\left[\min_{\lambda\ge 0}\max_{(\ell_t)_{t=1}^{\infty}\in L}\mathcal{L}(\ell_0,(\ell_t)_{t=1}^{\infty};\lambda)\right].
$$
Let $0\le \lambda^*\le \frac{1}{M-2}$, and $(\ell^{\lambda^*}_t)_{t=1}^{\infty}$ be the unique  maximizer\footnote{The maximizer is unique according to Lemma \ref{lem:M>2}.} of 
$
\mathbb{E}_{g_0}\sum_{t=1}^{\infty}\beta^t\left[F(c_t,\ell_t)+\lambda^* G(c_t,\ell_t)\right].
$
The inequality
\begin{equation}\label{eq:M>2cond}
    G(c_0^*,\ell_0^*)+\mathbb{E}_{g_0}\sum_{t=1}^{\infty}\beta^tG(c_t^{\lambda^*},\ell_t^{\lambda^*})> b_0u'(c_0^*)
\end{equation}
is satisfied.
\end{assumption}

Under this assumption, we can prove that the optimal $ \lambda $ is sufficiently small for the problem to be concave. The following result formalizes this.
\begin{proposition}\label{prop:largeM}
     Under Assumption \ref{ass:regu}, Assumption \ref{ass:v} with $M> 2$, and Assumption \ref{ass:M>2}, there is no lottery improvement.
\end{proposition}

\subsubsection{Example: CRRA utility}
To illustrate Assumption \ref{ass:M>2} we consider CRRA utility functions of the form
\[
u(c)=\frac{c^{1-\gamma}}{1-\gamma}\quad(\gamma>0),
\]
and assume that the labor utility function satisfies the first part of Assumption \ref{ass:v}. 

Lotteries are never optimal when $\underline{c}$ is sufficiently small.
For CRRA utility functions, the relative prudence and relative risk aversion are constant, i.e.,
$$
\frac{-u'''(c_t)c_t}{u''(c_t)}=\gamma+1,\quad\frac{-u''(c_t)c_t}{u'(c_t)}=\gamma.
$$

When $\gamma\le 1$, Assumption \ref{ass:v} with $M\le 2$ holds and Proposition \ref{prop:smallM} applies. 

When $\gamma>1$, let $(\ell^{\text{ref}}_t)_{t=0}^{\infty}$ be some feasible deterministic solution. Since $u(c)\rightarrow-\infty$ when $c\rightarrow0$ and since the optimal lottery should attain at least the utility of $(\ell^{\text{ref}}_t)_{t=0}^{\infty}$, we know that when $\underline{c}$ is sufficiently small, there exists some $\underline{\underline{c}}>0$ that is independent of $\underline{c}$, s.t.
$$
\arg\max_{\ell_0\in [\underline{c}+g_0,\bar{\ell}]}\left[\min_{\lambda\ge 0}\max_{(\ell_t)_{t=1}^\infty}\mathcal{L}(\ell_0,(\ell_t)_{t=1}^{\infty};\lambda)\right]\subset\left[\underline{\underline{c}}+g_0,\bar{\ell}\right], 
$$
which ensures that Assumption~\ref{ass:regu} holds when $\underline{c}$ is sufficiently small.
When $\lambda^*=\frac{1}{\gamma-1}$, the maximization problem for $(\ell_t)_{t=1}^{\infty}$ becomes
\begin{equation}\label{eq:max_l1_lambdastar}
\begin{aligned}
&\max_{(\ell_t)_{t=1}^{\infty}}\mathbb{E}_{g_0}\sum_{t=1}^{\infty}\beta^t\left[F(c_t,\ell_t)+\lambda^*G(c_t,\ell_t)\right]\\
=&\max_{(\ell_t)_{t=1}^{\infty}}\mathbb{E}_{g_0}\sum_{t=1}^{\infty}\beta^t\left[\underbrace{\left(u(c_t)+\lambda^*u'(c_t)c_t\right)}_{U(c_t)}+\underbrace{\left(v(\ell_t)+\lambda^*v'(\ell_t)\ell_t\right)}_{V(\ell_t)}\right]
\end{aligned}
\end{equation}
By straightforward computation, we have
$$
U'(c_t)=\lambda^*u''(c_t)c_t+(1+\lambda^*)u'(c_t)=\lambda^*u'(c_t)\left(\gamma-\frac{-u''(c_t)c_t}{u'(c_t)}\right)=0,
$$
and
$$
V'(\ell_t)=\lambda^*v''(\ell_t)\ell_t+(1+\lambda^*)v'(\ell_t)<0.
$$
Therefore, $U(c_t)+V(\ell_t)$ is decreasing in $\ell_t$, and the maximizer of problem \eqref{eq:max_l1_lambdastar} $(\ell^{\lambda^*}_t)_{t=1}^{\infty}$ satisfies
$$
\ell_t^{\lambda^*}(g^t)=\underline{c}+g_t,\quad\forall t\ge 1,\,g^t\in\mathcal{G}^t.
$$
When $\underline{c}\rightarrow0+$, we have
$$
G(c_0^*,\ell_0^*)+\mathbb{E}_{g_0}\sum_{t=1}^{\infty}\beta^tG(c_t^{\lambda^*},\ell_t^{\lambda^*})- b_0u'(c_0^*)\rightarrow\infty
$$
uniformly for $\ell_0^*$ in $[\underline{\underline{c}}+g_0,\bar{\ell}]$ and hence Assumption \ref{ass:M>2} is satisfied. Therefore, Proposition \ref{prop:largeM} applies for sufficiently small $\underline{c}>0$.

\subsubsection{Relation to \cite{citanna2024taxspots}}
According to the analysis in the previous section, one possible requirement for lottery improvement is that the function $U(c_t)$ in \eqref{eq:max_l1_lambdastar} is not concave even for very small $\lambda^*$. This is related to the fact that $u$ has large relative prudence at some consumption level. Indeed, \cite{citanna2024taxspots} perturb the CRRA utility function near the optimal deterministic consumption to
    $$
    u(c)=\frac{c^{1-\gamma}}{1-\gamma}+\chi(n(c-c^*))\frac{\zeta_c}{6}(c-c^*)^3,
    $$
    where $\chi\in C^{\infty}$ and 
    $$\chi(x)=
    \begin{cases}
    1,&-\frac{1}{2}\le x \le \frac{1}{2};\\
    0,&x\not\in[-1,1].
    \end{cases}
    $$
    This creates an extra term
    $$
    c\left(n\chi'(n(c-c^*))\frac{\zeta_c}{6}(c-c^*)^3+\underbrace{\chi(n(c-c^*))\frac{\zeta_c}{2}(c-c^*)^2}_{\text{convex  term}}\right)
    $$
    to $G(c_t,\ell_t)$. The convex term can be made arbitrarily convex by taking $\zeta_c$ sufficiently large, preventing $F(c_t,\ell_t)+\lambda G(c_t,\ell_t)$ from being concave for small $\lambda$. We consider a more natural class of utility functions where optimal lotteries arise.

\subsection{First best lotteries with unbounded relative risk aversion}

 In this subsection, we consider utility functions that exhibit unbounded relative risk aversion as consumption approaches its lower bound. We make the following assumption.
    \begin{assumption}\label{ass:lot_improve}
    The utility function for consumption $u:(\zeta,\infty) \rightarrow {\mathbb R}$ is increasing, concave, and 
    \begin{equation}\label{cond:lot_improve}
        \lim_{c\rightarrow\zeta+}u'(c)=\infty,\quad \lim_{c\rightarrow\zeta+}R(c):=\frac{-u''(c)c}{u'(c)}=\infty.
    \end{equation}
    \end{assumption}

    \begin{remark}
        We note that the Inada condition $\lim_{c\rightarrow\zeta+}u'(c)=\infty$ is implied by the unbounded relative risk aversion $\lim_{c\rightarrow\zeta+}R(c)=\infty$ when $\zeta=0$. Examples that satisfy Assumption \ref{ass:lot_improve} include shifted CRRA functions
        $$
        u(c)=\frac{(c-\zeta)^{1-\gamma}}{1-\gamma}\quad(\gamma>0,\ \zeta>0),
        $$
        and a negative exponential transformation of CRRA
        $$ u(c)=-e^{1/c^k} \quad (k > 0).$$
    \end{remark}
    
Under Assumption \ref{ass:lot_improve} the government can set the price of debt arbitrarily high by announcing a lottery next period that sets labor taxes close to 100 percent with a probability close to zero. It is easy to see that with shifted CRRA utility, the agent is forced to save against this outcome, independently of the interest rate. The probability can be chosen low enough to ensure that the effect on utility is second-order. In the limit, the government can then finance all its expenditure without ever taxing labor. 
The following result formalizes this idea.
    
    \begin{proposition}\label{prop:lot_improve}
        Under Assumption \ref{ass:regu}, Assumption \ref{ass:v}(\ref{enum_assv_1}) and Assumption \ref{ass:lot_improve}, the lottery problem achieves the first best when $\underline{c}\rightarrow \zeta+$, i.e.
        \begin{equation}\label{eq:goal_lot_improvement_primal}
        \lim_{\underline{c}\rightarrow\zeta} \max_{\underline{c}+g_0\le \ell_0\le \bar{\ell}}\min_{\lambda\ge 0}\max_{(\ell_t)_{t=1}^{\infty}\in L}\mathcal{L}(\ell_0,(\ell_t)_{t=1}^{\infty};\lambda)=\max_{(g_t+\zeta\le \ell_t\le \bar{\ell})_{t=0}^{\infty}}\mathbb{E}_{g_0}\sum_{t=0}^{\infty}\beta^tF(c_t,\ell_t).
        \end{equation}
    \end{proposition}

\subsection{Purification of lottery solutions}
The fact that lotteries are optimal does not necessarily imply that a deterministic solution cannot attain the same value. If $ \beta $ is sufficiently close to one, lotteries can be replaced with deterministic solutions that ``imitate'' the optimal lottery over time (as is well known in game theory; see, e.g., \citealp{fudenberg1991dispensability}).
To formalize this, we need the following assumptions, which guarantee that optimal lotteries have finite support.
\begin{assumption}\label{ass:utility}
 For any $g\in\mathcal{G}$, the utility functions $u(\ell-g)$ and $v(\ell)$ are both analytic functions on $[\underline{c}+g,\bar{\ell}]$.
\end{assumption}

    \begin{assumption}\label{ass:utility_regu}
       For any $g\in \mathcal{G}$, there are no $\lambda\in \mathbb{R}$ and $C\in\mathbb{R}$, such that
       $$
       F(\ell-g,\ell)+\lambda G(\ell-g,\ell)=C \ \forall\ell\in[\underline{c}+g,\bar{\ell}].
       $$
    \end{assumption}

The following lemma then establishes finite support. 
\begin{lemma}\label{lem:finite_support_strong}
     Under Assumption \ref{ass:utility} and \ref{ass:utility_regu}, for any $g\in\mathcal{G}$, there exists $M_g>0$, such that for any $\lambda\ge0$, $H_{\lambda}(\ell-g,\ell):=F(c,\ell)+\lambda G(c,\ell)$ has at most $M_g$ maximum points on $[\underline{c}+g,\bar{\ell}]$.   
    \end{lemma}
  
    \begin{remark}\label{rmk:ass_property}
      $M_g$ is a number that depends only on the utility functions $u(\ell-g)$ and $v(\ell)$ over the interval $[\underline{c}+g,\bar{\ell}]$. It can be bounded in terms of the numbers of zeros in $[\underline{c}+g,\bar{\ell}]$ of
\[
G'(\ell-g,\ell)
\quad\text{and}\quad
F''(\ell-g,\ell)G'(\ell-g,\ell)-G''(\ell-g,\ell)F'(\ell-g,\ell).
\]
In most applications, these two expressions can either be analytically extended to $\zeta$ or have a nonzero limit at $\zeta$; in either case, $M_g$ is also independent of $\underline{c}$.
    \end{remark}
 Finally, we make a simple regularity assumption on the Markov process of government expenditure.
    \begin{assumption}\label{ass:trans}
        For any $g,g'\in\mathcal{G}$, $\pi(g'|g)>0$. We define $\pi_{\min}=\min_{g,g'\in\mathcal{G}}\pi(g'|g)$.
    \end{assumption}
We now formally define purification.
\begin{definition}\label{def:purification}
    Suppose that $L=\{(\ell_t(g^t):\underline{c}+g_t\le \ell_t(g^t)\le \bar{\ell},\,\forall t\ge 1,\,g^t\in\mathcal{G}^t\}$. We say a lottery $P\in\mathcal{P}(L)$ can be purified, if there exists $(\ell_t(g^t))\in L$, s.t.
     $$
\sum_{t=1}^{\infty}\sum_{g^t\in\mathcal{G}^t}\beta^t\pi(g^t|g_0)F(\ell_t(g^t)-g_t,\ell_t(g^t))=\mathbb{E}_{g_0}^{(\ell_t)_{t \ge 1}\sim P}\sum_{t=1}^{\infty}\beta^tF(\ell_t-g_t,\ell_t),
    $$
    and
    $$
\sum_{t=1}^{\infty}\sum_{g^t\in\mathcal{G}^t}\beta^t\pi(g^t|g_0)G(\ell_t(g^t)-g_t,\ell_t(g^t))=\mathbb{E}_{g_0}^{(\ell_t)_{t \ge 1}\sim P}\sum_{t=1}^{\infty}\beta^tG(\ell_t-g_t,\ell_t).
    $$
\end{definition}
Our main theorem is as follows.
\begin{theorem}\label{thm:purification}
    Under Assumption \ref{ass:utility}, Assumption \ref{ass:utility_regu}, Assumption \ref{ass:trans}, and Assumption \ref{ass:regu}, when 
    $$
    1>\beta\ge \max_{g\in\mathcal{G}}\frac{M_g-1}{M_g-1+\pi_{\min}},
    $$
    where $M_g$ is defined in Lemma \ref{lem:finite_support_strong}, any lottery solution to \eqref{eq:Lucas-S-l} can be purified.
\end{theorem}
To prove the theorem,
    assume that $\ell_0^*$, $P^*\in\mathcal{P}(L)$ is a lottery solution to \eqref{eq:Lucas-S-l}. 
    According to Corollary 2.1 in \cite{shen2025lagrangian}, every $(\ell^*_t)_{t\ge 1}$ on the support of $P^*$ maximizes
    $$
    \mathbb{E}_{g_0}\sum_{t=1}^{\infty}\beta^tH_{\lambda^*}(\ell_t-g_t,\ell_t)
    $$
    for some $\lambda^*\ge 0$. According to Lemma \ref{lem:finite_support_strong}, there exists $ L^{\text{lot}}(g):=\{\ell^{*,1}(g),\cdots, \ell^{*,M_g}(g)\}=\arg\max_{\ell\in [\underline{c}+g,\bar{\ell}]}H_{\lambda^*}(\ell-g,\ell)$ for each $g\in\mathcal{G}$, such that $P^*(\ell_t(g^t)\in L^{\text{lot}}(g_t))=1,\,\,\forall t\ge 1,\,g^t\in\mathcal{G}^t.$ Therefore, for
    each $t\ge 1,\,g^t\in\mathcal{G}^t$, $P^*$ has a marginal distribution $P^*_{t,g^t}\in\mathcal{P}(L^{\text{lot}}(g))$, and we have
    \begin{equation}\label{eq:purification_mid1}
    \begin{aligned}
    &\mathbb{E}_{g_0}^{(\ell_t)_{t \ge 1}\sim P^*}\sum_{t=1}^{\infty}\beta^tF(\ell_t-g_t,\ell_t)\\
    =&\sum_{t=1}^{\infty}\sum_{g^t\in\mathcal{G}^t}\mathbb{E}^{(\ell_t)_{t \ge 1}\sim P^*}\beta^t\pi(g^t|g_0)F(\ell_t(g^t)-g_t,\ell_t(g^t))\\
    =&\sum_{t=1}^{\infty}\sum_{g^t\in\mathcal{G}^t}\mathbb{E}^{\ell_t(g^t)\sim P^*_{t,g^t}}\beta^t\pi(g^t|g_0)F(\ell_t(g^t)-g_t,\ell_t(g^t))\\
    =&\sum_{t=1}^{\infty}\sum_{g^t\in\mathcal{G}^t}\sum_{k=1}^{M_{g_t}}P^*_{t,g^t}(\ell^{*,k}(g_t))\beta^t\pi(g^t|g_0)F(\ell^{*,k}(g_t)-g_t,\ell^{*,k}(g_t)).
    \end{aligned}
    \end{equation}
    For each $g\in\mathcal{G}$, we define $P_g^*\in\mathcal{P}(L^{\text{lot}}(g))$, s.t.
    $$
    P_g^*(\ell^{*,k}(g))=\left[\sum_{t=1}^{\infty}\sum_{g^t\in\mathcal{G}^t\text{ s.t. }g_t=g}\beta^t\pi(g^t|g_0)\right] ^{-1}\sum_{t=1}^{\infty}\sum_{g^t\in\mathcal{G}^t\text{ s.t. }g_t=g}\beta^t\pi(g^t|g_0)P^*_{t,g^t}(\ell^{*,k}(g_t)),
    $$
    and we have
\begin{equation}\label{eq:purification_mid2}
\begin{aligned}
    &\sum_{t=1}^{\infty}\sum_{g^t\in\mathcal{G}^t}\sum_{k=1}^{M_{g_t}}P^*_{t,g^t}(\ell^{*,k}(g_t))\beta^t\pi(g^t|g_0)F(\ell^{*,k}(g_t)-g_t,\ell^{*,k}(g_t))\\
    =&\sum_{g\in\mathcal{G}}\sum_{k=1}^{M_{g}}\sum_{t=1}^{\infty}\sum_{g^t\in\mathcal{G}^t\text{ s.t. }g_t=g}P^*_{t,g^t}(\ell^{*,k}(g))\beta^t\pi(g^t|g_0)F(\ell^{*,k}(g)-g,\ell^{*,k}(g))\\
    =&  \sum_{g\in\mathcal G}\left[\sum_{t=1}^{\infty}\sum_{g^t\in\mathcal{G}^t\text{ s.t. }g_t=g}\beta^t\pi(g^t|g_0)\right] \cdot\left[\sum_{k=1}^{M_g}P^*_g(k)F(\ell^{*,k}(g)-g,\ell^{*,k}(g))\right].\end{aligned}
    \end{equation}
    We combine \eqref{eq:purification_mid1} and \eqref{eq:purification_mid2} and have
\begin{equation}\label{eq:pindownF}
    \sum_{g\in\mathcal G}\left[\sum_{t=1}^{\infty}\sum_{g^t\in\mathcal{G}^t\text{ s.t. }g_t=g}\beta^t\pi(g^t|g_0)\right] \cdot\left[\sum_{k=1}^{M_g}P^*_g(k)F(\ell^{*,k}(g)-g,\ell^{*,k}(g))\right]=\mathbb{E}_{g_0}^{(\ell_t)_{t \ge 1}\sim P^*}\sum_{t=1}^{\infty}\beta^tF(\ell_t-g_t,\ell_t).
    \end{equation}
    Similarly, 
    \begin{equation}\label{eq:pindownG}
    \sum_{g\in\mathcal G}\left[\sum_{t=1}^{\infty}\sum_{g^t\in\mathcal{G}^t\text{ s.t. }g_t=g}\beta^t\pi(g^t|g_0)\right] \cdot\left[\sum_{k=1}^{M_g}P^*_g(k)G(\ell^{*,k}(g)-g,\ell^{*,k}(g))\right]=\mathbb{E}_{g_0}^{(\ell_t)_{t \ge 1}\sim P^*}\sum_{t=1}^{\infty}\beta^tG(\ell_t-g_t,\ell_t).
    \end{equation}
    The lottery solution $P^*$ is then pinned down to $|\mathcal{G}|$ finite-support probability distributions $P^*_g\in\mathcal{P}(L^{\text{lot}}(g))$.
Figure~\ref{fig:interval_decomposition} illustrates the intuition for the case where $\mathcal{G}=\{g\}$ is a singleton and $L^{\text{lot}}(g)=\{\ell^{*,1}(g),\ell^{*,2}(g)\}$.
\begin{figure}[htp]
\centering
\footnotesize
\begin{tikzpicture}[xscale=13,yscale=1.5]
\def\h{0.6}

\filldraw[fill=blue!40,draw=black] (0,0) rectangle (0.2,\h);
\filldraw[fill=red!50,draw=black] (0.2,0) rectangle (0.5,\h);

\filldraw[fill=blue!40,draw=black] (0.5,0) rectangle (0.65,\h);
\filldraw[fill=red!50,draw=black] (0.65,0) rectangle (0.75,\h);

\filldraw[fill=blue!40,draw=black] (0.75,0) rectangle (0.8,\h);
\filldraw[fill=red!50,draw=black] (0.8,0) rectangle (0.875,\h);

\filldraw[fill=blue!40,draw=black] (0.875,0) rectangle (0.92,\h);
\filldraw[fill=red!50,draw=black] (0.92,0) rectangle (0.9375,\h);

\filldraw[fill=blue!40,draw=black] (0.9375,0) rectangle (0.95,\h);
\filldraw[fill=red!50,draw=black] (0.95,0) rectangle (0.96875,\h);

\node at (0.985,0.3) {$\cdots$};

\foreach \x in {0.5, 0.75, 0.875, 0.9375} {
  \draw[dashed,gray] (\x,0) -- (\x,\h);
}

\foreach \x in {0, 0.5, 0.75, 0.875, 0.9375, 1} {
  \draw (\x,\h) -- (\x,\h+0.08);
}
\node[above] at (0,\h+0.08) {\footnotesize $0$};
\node[above] at (0.5,\h+0.08) {\footnotesize $1-\beta$};
\node[above] at (0.75,\h+0.08) {\footnotesize $1-\beta^2$};
\node[above] at (0.875,\h+0.08) {\footnotesize $1-\beta^3$};
\node[above] at (1,\h+0.08) {\footnotesize $1$};

\node[above,font=\scriptsize] (lab11) at (0.1,\h+0.45)
  {$(1-\beta)P_1^*(\ell^{*,1}(g))$};
\draw[gray,thin] (0.1,\h) -- (lab11.south);

\node[above,font=\scriptsize] (lab12) at (0.35,\h+0.45)
  {$(1-\beta)P_1^*(\ell^{*,2}(g))$};
\draw[gray,thin] (0.35,\h) -- (lab12.south);

\node[above,font=\scriptsize] (lab21) at (0.48,\h+1.0)
  {$(1-\beta)\beta P_2^*(\ell^{*,1}(g))$};
\draw[gray,thin] (0.575,\h) -- (lab21.south);

\node[above,font=\scriptsize] (lab22) at (0.74,\h+1.0)
  {$(1-\beta)\beta P_2^*(\ell^{*,2}(g))$};
\draw[gray,thin] (0.7,\h) -- (lab22.south);

\draw[decorate,decoration={brace,amplitude=4pt,mirror},thick] (0,-0.05) -- (0.5,-0.05)
  node[midway,below=5pt] {\small $I_1$};
\draw[decorate,decoration={brace,amplitude=4pt,mirror},thick] (0.5,-0.05) -- (0.75,-0.05)
  node[midway,below=5pt] {\small $I_2$};
\draw[decorate,decoration={brace,amplitude=4pt,mirror},thick] (0.75,-0.05) -- (0.875,-0.05)
  node[midway,below=5pt] {\small $I_3$};
\draw[decorate,decoration={brace,amplitude=4pt,mirror},thick] (0.875,-0.05) -- (0.9375,-0.05)
  node[midway,below=5pt] {\small $I_4$};
\draw[decorate,decoration={brace,amplitude=4pt,mirror},thick] (0.9375,-0.05) -- (0.96875,-0.05)
  node[midway,below=5pt] {\small $I_5$};

\begin{scope}[shift={(0.55,\h+0.45)}]
  \filldraw[fill=blue!40,draw=black] (0,0) rectangle (0.06,0.12);
  \node[right] at (0.07,0.06) {\small $\ell^{*,1}(g)$};
  \filldraw[fill=red!50,draw=black] (0.28,0) rectangle (0.34,0.12);
  \node[right] at (0.35,0.06) {\small $\ell^{*,2}(g)$};
\end{scope}
\end{tikzpicture}
\caption{The interval $[0,1]$ is decomposed into countably many subintervals $I_k$ of length $|I_k|=(1-\beta)\beta^{k-1}$, $k\ge1$, so that $\sum_{k\ge1}|I_k|=1$. On each $I_k$, the left part (blue) is assigned to $\ell^{*,1}(g)$ and the right part (red) to $\ell^{*,2}(g)$. The distribution $P^*_g$, which represents the total mass assigned to the blue and red parts across all periods, is given by
\(
P^*_g=\sum_{t=1}^{\infty}(1-\beta)\beta^{t-1}P_t^*.
\)}
\label{fig:interval_decomposition}
\end{figure}
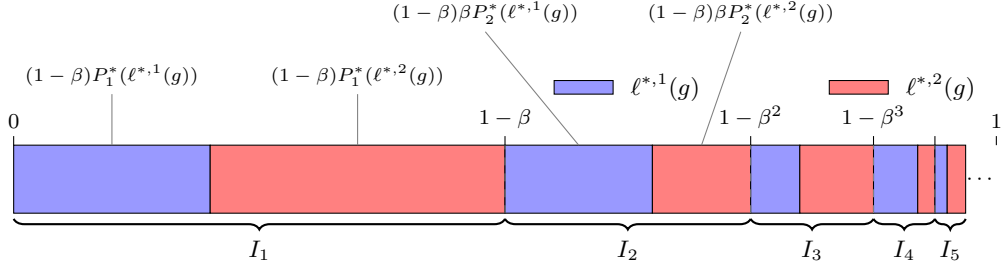

We now turn to purification. Consider again the simple case where $\mathcal{G}=\{g\}$ is a singleton and $L^{\text{lot}}(g)=\{\ell^{*,1}(g),\ell^{*,2}(g)\}$. The purification problem asks whether, instead of splitting each interval $I_k$ between the two outcomes, one can assign each $I_k$ entirely to one of $\ell^{*,1}(g)$ and $\ell^{*,2}(g)$, while still generating the same distribution $P_g^*$. 
Lemma 1 in \cite{fudenberg1991dispensability} gives an affirmative answer. The following lemma states this in our stochastic setting.

    \begin{lemma}\label{lem:purify}
        Assume that $M_g\in\mathbb{N}_+$,
        $$
        1>\beta\ge\frac{M_g-1}{M_g-1+\pi_{\min}}.
        $$
        For any $\vec{p}\in\Delta^{M_g}_+:=\{\vec{p}\in\mathbb{R}_+^{M_g}:\sum_{k=1}^{M_g}p_k=1\}$, there exists $(k_t(g^t)\in\{1,\cdots, M_g\})_{t\ge 1,\,g^t\in\mathcal{G}^t\text{ s.t. }g_t=g}$, s.t.
        $$
        \frac{1}
{\sum_{s=1}^{\infty}\sum_{h^s\in\mathcal{G}^s:\,h_s=g}\beta^s\pi(h^s\mid g_0)} \sum_{t=1}^{\infty}\sum_{g^t\in\mathcal{G}^t:\,g_t=g}\left[\beta^t\pi(g^t|g_0)\right]\vec{e}_{k_t(g^t)}=\vec{p},
        $$
        where 
\[
\vec{e}_{k}=(0,\ldots,0,\underset{\substack{\uparrow\\ k}}{1},0,\ldots,0)
\]
denotes the unit vector whose $k$-th component is $1$.
    \end{lemma}
  
The proof of the theorem then follows directly from  Lemma \ref{lem:purify} with $\vec{p}=P^*_g$ for each $g\in\mathcal{G}$. Indeed, for any $g\in\mathcal{G}$, $t\ge 1,\,g^t\in\mathcal{G}^t$ s.t. $g_t=g$, one can take $\ell_t(g^t)=\ell^{*,k_t(g^t)}(g)$. Therefore,
$$
\begin{aligned}
&\sum_{t=1}^{\infty}\sum_{g^t\in\mathcal{G}^t\text{ s.t. }g_t=g}\beta^t\pi(g^t|g_0)F(\ell_t(g^t)-g_t,\ell_t(g^t))\\
=&\left[\sum_{t=1}^{\infty}\sum_{g^t\in\mathcal{G}^t\text{ s.t. }g_t=g}\beta^t\pi(g^t|g_0)\right]\cdot\left[\sum_{k=1}^{M_g}P^*_g(k)F(\ell^{*,k}(g)-g,\ell^{*,k}(g))\right],
\end{aligned}
$$
and
$$
\begin{aligned}
&\sum_{t=1}^{\infty}\sum_{g^t\in\mathcal{G}^t\text{ s.t. }g_t=g}\beta^t\pi(g^t|g_0)G(\ell_t(g^t)-g_t,\ell_t(g^t))\\
=&\left[\sum_{t=1}^{\infty}\sum_{g^t\in\mathcal{G}^t\text{ s.t. }g_t=g}\beta^t\pi(g^t|g_0)\right]\cdot\left[\sum_{k=1}^{M_g}P^*_g(k)G(\ell^{*,k}(g)-g,\ell^{*,k}(g))\right],
\end{aligned}
$$
for each $g\in\mathcal{G}$.
Hence we have
$$
\begin{aligned}
&\sum_{t=1}^{\infty}\sum_{g^t\in\mathcal{G}^t}\beta^t\pi(g^t|g_0)F(\ell_t(g^t)-g_t,\ell_t(g^t))\\
=&\sum_{g\in\mathcal{G}}\sum_{t=1}^{\infty}\sum_{g^t\in\mathcal{G}^t\text{ s.t. }g_t=g}\beta^t\pi(g^t|g_0)F(\ell_t(g^t)-g_t,\ell_t(g^t))\\
=&\sum_{g\in\mathcal G}\left[\sum_{t=1}^{\infty}\sum_{g^t\in\mathcal{G}^t\text{ s.t. }g_t=g}\beta^t\pi(g^t|g_0)\right] \cdot\left[\sum_{k=1}^{M_g}P^*_g(k)F(\ell^{*,k}(g)-g,\ell^{*,k}(g))\right]\\=&\mathbb{E}_{g_0}^{(\ell_t)_{t \ge 1}\sim P^*}\sum_{t=1}^{\infty}\beta^tF(\ell_t-g_t,\ell_t),
\end{aligned}
$$
where the last equality follows from \eqref{eq:pindownF}. Similarly
$$
\begin{aligned}
   &\sum_{t=1}^{\infty}\sum_{g^t\in\mathcal{G}^t}\beta^t\pi(g^t|g_0)G(\ell_t(g^t)-g_t,\ell_t(g^t))\\
=&\mathbb{E}_{g_0}^{(\ell_t)_{t \ge 1}\sim P^*}\sum_{t=1}^{\infty}\beta^tG(\ell_t-g_t,\ell_t). 
\end{aligned}
$$

\subsection{Possible failures of purification}
For small values of $ \beta $, or with explicit debt constraints, it might not be possible to purify the optimal lottery, and the lottery solution strictly dominates the deterministic one. We provide two examples to illustrate these cases.

\subsubsection{Essential lotteries for small $ \beta$}
\label{app:cm_with_lotteries}
We consider a problem with $\mathcal{G}=\{0\}$, i.e.
\begin{equation}\label{eq:reduced_Ramsey_lot_imp}
\begin{aligned}
    &\max_{\ell_0,P\in\mathcal{P}(L)} \mathbb{E}^{P}\sum_{t=0}^{\infty}\beta^t[u(\ell_t)+v(\ell_t)]\\
    \text{s.t. }&\mathbb{E}^P\sum_{t=0}^{\infty}\beta^t[u'(\ell_t)\ell_t+v'(\ell_t)\ell_t]\ge b_0u'(\ell_0),
\end{aligned}
\end{equation}
where $u(\ell_t)=\log (\ell_t-0.1),\,v(\ell_t)=\sqrt{1-\ell_t}$. Since government expenditure is $g=0$, consumption always equals labor supply. The first-best labor supply (and consumption) is 
$$
\ell^{\text{FB}}\text{ s.t. }u'(\ell^{\text{FB}})+v'(\ell^{\text{FB}})=0\Rightarrow \ell^{\text{FB}}\approx 0.8568,\,u(\ell^{\text{FB}})+v(\ell^{\text{FB}})\approx0.0998,\,u'(\ell^{\text{FB}})\approx1.3213.
$$
Consumption is bounded below by some $\underline{c}>0.1$. Assume that $\beta=0.2$ and that the initial debt $b_0>0$ is such that 
$$
\footnotesize
\begin{aligned}
b_0=&\frac{1}{u'(\ell^{\text{FB}})}\left\{\sum_{t\ge 0,\,t\ne1}\beta^t\left[u'(\ell^{\text{FB}})\ell^{\text{FB}}+v'(\ell^{\text{FB}})\ell^{\text{FB}}\right]+\beta \left[\frac{1}{2}(u'(\underline{c})\underline{c}+v'(\underline{c})\underline{c})+\frac{1}{2}\left(u'(\ell^{\text{FB}})\ell^{\text{FB}}+v'(\ell^{\text{FB}})\ell^{\text{FB}}\right)\right]\right\}\\
=&\frac{1}{u'(\ell^{\text{FB}})}\left\{\frac{\beta}{2}(u'(\underline{c})\underline{c}+v'(\underline{c})\underline{c})\right\}.
\end{aligned}
$$
We will show that, when $\underline{c}\rightarrow 0.1+$, the lottery solution attains strictly higher welfare than the deterministic solution.
\paragraph{The lottery solution}
A feasible lottery solution is 
$$
\ell_t=\ell^{\text{FB}},\,\forall t\ne1;\, \ell_1^H=\ell^{\text{FB}}(p=0.5),\,\ell_1^L=\underline{c}(p=0.5).
$$
We define $V^{\text{FB}}$ as the first best welfare. The total welfare of this lottery solution is
$$
V^{\text{lot}}= V^{\text{FB}}-\frac{\beta}{2} \left\{\underbrace{\left[u(\ell^{\text{FB}})+v(\ell^{\text{FB}})\right]}_{\approx0.0998<1}-[\underbrace{u(\underline{c})}_{=\log(\underline{c}-0.1)}+\underbrace{v(\underline{c})}_{>0}]\right\}>\underbrace{V^{\text{FB}}-0.1+0.1\log(\underline{c}-0.1)}_{\underline{V}^{\text{lot}}}.
$$
\paragraph{The deterministic solution}
All feasible deterministic solutions must satisfy
$$
b_0\le \frac{1}{u'(\ell_0)}\left\{\sum_{t\ge 0}\beta^t\left[u'(\ell_t)\ell_t+v'(\ell_t)\ell_t\right]\right\},
$$
implying that
$$
\begin{aligned}
&\sum_{t\ge 0}\beta^t\left[u'(\ell_t)\ell_t+v'(\ell_t)\ell_t\right]\ge \frac{u'(\ell_0)}{u'(\ell^{\text{FB}})}\left[\frac{\beta}{2}\left(u'(\underline{c})\underline{c}+v'(\underline{c}\right)\underline{c})\right]\\\ge& \underbrace{\frac{u'(\bar{\ell})}{u'(\ell^{\text{FB}})}}_{\approx \frac{1}{1.3213(\bar{\ell}-0.1)}>\frac{2}{3}}\left[\frac{\beta}{2}\left(u'(\underline{c})\underline{c}+v'(\underline{c}\right)\underline{c})\right]\ge \frac{\beta}3\left(u'(\underline{c})\underline{c}+v'(\underline{c})\underline{c}\right).
\end{aligned}
$$
Since
$$
\sum_{t\ge 2}\beta^t\left[u'(\ell_t)\ell_t+v'(\ell_t)\ell_t\right]\le \frac{\beta^2}{1-\beta}\left(u'(\underline{c})\underline{c}+v'(\underline{c})\underline{c}\right)=\frac{\beta}{4}\left(u'(\underline{c})\underline{c}+v'(\underline{c})\underline{c}\right),
$$
we have
\begin{equation}\label{tech_cond}
\sum_{t=0}^{1}\beta^t\left[u'(\ell_t)\ell_t+v'(\ell_t)\ell_t\right]\ge \frac{1}{12}\beta\left(u'(\underline{c})\underline{c}+v'(\underline{c})\underline{c}\right)=\underbrace{\frac{1}{12}\beta\left(\frac{\underline{c}}{\underline{c}-0.1}-\frac{\underline{c}}{2\sqrt{1-\underline{c}}}\right)}_{I}.
\end{equation}
We now show that there exists some $\underline{\underline{c}}>0.1$, s.t. when $\underline{c}<\underline{\underline{c}}$, we have
\begin{equation}\label{eq:tech_example_cond}
\ell_0\le 73(\underline{c}-0.1)+0.1,\,\textbf{or }\ell_1\le 73(\underline{c}-0.1)+0.1. 
\end{equation}
We prove this by contradiction. Suppose the contrary. Then
$$
\sum_{t=0}^{1}\beta^t\left[u'(\ell_t)\ell_t+v'(\ell_t)\ell_t\right]\le \underbrace{6\beta\left[\frac{\underline{c}+72(\underline{c}-0.1)}{73(\underline{c}-0.1)}-\frac{73(\underline{c}-0.1)+0.1}{2\sqrt{0.9-73(\underline{c}-0.1)}}\right]}_{II}
$$
We know that
$$
\lim_{\underline{c}\rightarrow0.1+}\frac{II}{I}=\frac{72}{73}<1,
$$
so \eqref{tech_cond} fails as $\underline{c}\to 0.1^+$, a contradiction. Hence, \eqref{eq:tech_example_cond} holds. Therefore, the welfare of any feasible deterministic solution satisfies
$$
\begin{aligned}
V^{\det}\le& V^{\text{FB}}-\beta \left\{\underbrace{\left[u(\ell^{\text{FB}})+v(\ell^{\text{FB}})\right]}_{\approx 0.0998>0}-[\underbrace{u(73(\underline{c}-0.1)+0.1)}_{=\log(73(\underline{c}-0.1))}+\underbrace{v(73(\underline{c}-0.1)+0.1)}_{<1}]\right\}\\
\le &V^{\text{FB}}+0.2\log(73(\underline{c}-0.1))+0.2=\underbrace{V^{\text{FB}}+0.2\times (1+\log 73)+0.2\log (\underline{c}-0.1)}_{\bar{V}^{\det}}.
\end{aligned}
$$
It is straightforward to see that $\bar{V}^{\det}<\underline{V}^{\text{lot}}$ when $\underline{c}\rightarrow 0.1+$.

  \subsubsection{Debt constraints}
Purification relies crucially on debt constraints not being too tight. We show this using an example that satisfies Assumption~\ref{ass:lot_improve}. Assume that $b_0>0$ and the debt constraint is $\bar{b}=\max\{1,b_0\}$.
\paragraph{The lottery solution}
For any $g\in\mathcal{G}$, define
\[
\small
V_F^{\mathrm{FB}}(g)
=\mathbb{E}\left[\sum_{t=0}^{\infty}\beta^t
F\bigl(\ell^{\mathrm{FB}}(g_t)-g_t,\ell^{\mathrm{FB}}(g_t)\bigr)
\,\middle|\,g_0=g\right],\,V_G^{\mathrm{FB}}(g)
=\mathbb{E}\left[\sum_{t=0}^{\infty}\beta^t
G\bigl(\ell^{\mathrm{FB}}(g_t)-g_t,\ell^{\mathrm{FB}}(g_t)\bigr)
\,\middle|\,g_0=g\right],
\]
where $\ell^{\mathrm{FB}}(g)$ is determined by
\[
u'(\ell^{\mathrm{FB}}(g)-g)+v'(\ell^{\mathrm{FB}}(g))=0.
\]
Consequently,
\[
G\bigl(\ell^{\mathrm{FB}}(g)-g,\ell^{\mathrm{FB}}(g)\bigr)
=-g\,u'(\ell^{\mathrm{FB}}(g)-g)\le 0.
\]
Let $g_{\min}$ be the smallest element of $\mathcal{G}$. We consider the following lottery:
$$
\begin{aligned}
&\ell^L_1(g_0,g_{\min})=\underline{c}+g_{\min}(\text{with probability } p),\,\ell_1^H(g_0,g_{\min})=\ell^{\text{FB}}(g_{\min})(\text{with probability } 1-p);\\
&\ell_t(g^t)=\ell^{\text{FB}}(g_t)\text{ otherwise},
\end{aligned}
$$
where
$$
p=\frac{u'(\ell^{\text{FB}}(g_0)-g_0)b_0-V_G^{\text{FB}}}{\beta \pi(g_{\min}|g_0)\left[G\bigl(\underline{c},\underline{c}+g_{\min}\bigr)-G\bigl(\ell^{\mathrm{FB}}(g_{\min})-g,\ell^{\mathrm{FB}}(g_{\min})\bigr)\right]}\rightarrow0(\underline{c}\rightarrow \zeta).
$$
It is straightforward to check that
\begin{itemize}
    \item The constraint is satisfied;
    \item The total welfare is  $V_F^{\text{FB}}(g_0)+\beta\pi(g_{\min}|g_0)p\left[F\bigl(\underline{c},\underline{c}+g_{\min}\bigr)-F\bigl(\ell^{\mathrm{FB}}(g_{\min})-g_{\min},\ell^{\mathrm{FB}}(g_{\min})\bigr)\right]\rightarrow V_{F}^{\text{FB}}(g_0)$ as $\underline{c}\rightarrow \zeta$\footnote{This relies on the fact that $F(\underline{c},\underline{c}+g_{\min})/G(\underline{c},\underline{c}+g_{\min})\rightarrow0$, which is established in the proof of Proposition \ref{prop:lot_improve}.}.
    \item The debts $(b_t)_{t\ge 1}$ satisfy
    $$
    b^L_t(g_0,g_{\min})=\underline{c}+\frac{1}{u'(\underline{c})}\left[v'(\underline{c}+g_{\min})(\underline{c}+g_{\min})+\sum_{g\in\mathcal{G}}\beta\pi(g|g_{\min})V^{FB}_G(g)\right]\le \underline{c}<\bar{b}; \,b_t\le 0\text{ otherwise.}
    $$
\end{itemize}

\paragraph{The deterministic solution}
For each $g\in\mathcal{G}$, $u'(\ell-g)$ and $u'(\ell-g)(\ell-g)+v'(\ell)\ell$ are continuous at $\ell^{\text{FB}}(g)$. Hence for any $\epsilon>0$, there exist $T_1(\epsilon)>0$ and $T_2(\epsilon)>0$, such that for each $g\in\mathcal{G}$, whenever $|\ell-\ell^{\text{FB}}(g)|\le \epsilon$,
\[
\left|\left[u'(\ell-g)(\ell-g)+v'(\ell)\ell\right]-\underbrace{\left[u'(\ell^{\text{FB}}(g)-g)(\ell^{\text{FB}}(g)-g)+v'(\ell^{\text{FB}}(g))\ell^{\text{FB}}(g)\right]}_{\le 0}\right|\le T_1(\epsilon)
\]
and
\[
|u'(\ell-g)-u'(\ell^{\text{FB}}(g)-g)|\le T_2(\epsilon),
\]
and that $T_1(\epsilon),T_2(\epsilon)\to 0$ as $\epsilon\to 0$. At each state we have
\[
b_t(g^t)u'(c_t(g^t))\le u'(c_t(g^t))c_t(g^t)+v'(\ell_t(g^t))\ell_t(g^t)+\beta\mathbb{E}_{g_t}b_{t+1}(g^{t+1})u'(c_{t+1}(g^{t+1})).
\]
We define $c^{\text{FB}}(g)=\ell^{\text{FB}}(g)-g$, and take $\epsilon$ small enough s.t.
$$
b_0u'(c^{\text{FB}}(g_0))>
\frac{(1+\beta)\bar{b}T_2(\epsilon)+ T_1(\epsilon)}{1-\beta}.
$$
We now show that there exists $\bar{T}>0$ such that, for some $t\le\bar{T}+2$ and some $g^t\in\mathcal{G}^t$,
\[
|\ell_t(g^t)-\ell^{\text{FB}}(g_t)|>\epsilon.
\]
Assume that $|\ell_t(g^t)-\ell^{\text{FB}}(g_t)|\le \epsilon$ for all $t\le T+1$. We now prove by induction that we can find a sequence $g^0\prec g^1\prec\cdots\prec g^{T+1}$, s.t. 
\begin{equation}\label{eq:CMup_ind1}
b_t(g^t)u'(c^{\text{FB}}(g_t))>\frac{(1+\beta)\bar{b}T_2(\epsilon)+ T_1(\epsilon)}{1-\beta},
\end{equation}
and
\begin{equation}\label{eq:CMup_ind2}
\begin{aligned}
b_{t+1}(g^{t+1})u'(c^{\text{FB}}(g_t))\ge& \frac{1}{\beta}b_t(g^t)u'(c^{\text{FB}}(g_t))-\frac{1+\beta}{\beta}\bar{b}T_2(\epsilon)-\frac{1}{\beta}T_1(\epsilon)\\
=&\frac{1}{\beta}\left[b_t(g^t)u'(c^{\text{FB}}(g^t))-\frac{(1+\beta)\bar{b}T_2(\epsilon)+ T_1(\epsilon)}{1-\beta}\right]+\frac{(1+\beta)\bar{b}T_2(\epsilon)+ T_1(\epsilon)}{1-\beta}.
\end{aligned}.
\end{equation}
Indeed, if for some $t\ge 0,\,g^t\in\mathcal{G}^t$, \eqref{eq:CMup_ind1} is satisfied, then we have
$$
\begin{aligned}
&\mathbb{E}_{g_t}b_{t+1}(g^{t+1})u'(c_{t+1}(g^{t+1}))\\\ge& \frac{1}{\beta}\left[b_t(g^t)u'(c_t(g^t))-T_1(\epsilon)\right]\\\ge& \frac{1}{\beta}\left[b_t(g^t)[u'(c^{\text{FB}}(g_t))-T_2(\epsilon)]-T_1(\epsilon)\right]\\\ge& \frac{1}{\beta}b_t(g^t)u'(c^{\text{FB}}(g_t))-\frac{1}{\beta}\bar{b}T_2(\epsilon)-\frac{1}{\beta}T_1(\epsilon).
\end{aligned}
$$
Therefore, there exists some $g^t\prec g^{t+1}$, s.t.
$$
b_{t+1}(g^{t+1})u'(c_{t+1}(g^{t+1}))\ge \frac{1}{\beta}b_t(g^t)u'(c^{\text{FB}}(g_t))-\frac{1}{\beta}\bar{b}T_2(\epsilon)-\frac{1}{\beta}T_1(\epsilon)>0.
$$
It follows that $b_{t+1}(g^{t+1})>0$, and
\[
\begin{aligned}
    &b_{t+1}(g^{t+1})u'(c^{\mathrm{FB}}(g_t))\\
    \ge{}& b_{t+1}(g^{t+1})u'(c_{t+1}(g^{t+1}))-\bar{b}T_2(\epsilon)\\
    \ge{}& \frac{1}{\beta}b_t(g^t)u'(c^{\mathrm{FB}}(g_t))
          -\frac{1+\beta}{\beta}\bar{b}T_2(\epsilon)
          -\frac{1}{\beta}T_1(\epsilon),
\end{aligned}
\]
which is precisely \eqref{eq:CMup_ind2}. It is then straightforward to verify that
$$
b_{t+1}(g^{t+1})u'(c^{\text{FB}}(g_t))>\frac{(1+\beta)\bar{b}T_2(\epsilon)+ T_1(\epsilon)}{1-\beta}
$$
according to \eqref{eq:CMup_ind1} and \eqref{eq:CMup_ind2}.

Combining the inequalities \eqref{eq:CMup_ind2} for $g^0, g^1, \dots, g^T$, we obtain
$$
b_{T+1}(g^{T+1})u'(c^{\text{FB}}(g_{T+1}))\ge (\frac{1}{\beta})^{T}\left[b_0u'(c^{\text{FB}}(g_0))-\frac{(1+\beta)\bar{b}T_2(\epsilon)+ T_1(\epsilon)}{1-\beta}\right]+\frac{(1+\beta)\bar{b}T_2(\epsilon)+ T_1(\epsilon)}{1-\beta},
$$
or equivalently
\begin{equation}\label{eq:incomplete_upper}
b_{T+1}(g^{T+1})\ge (u'(c^{\text{FB}}(g_{T+1})))^{-1}\left\{(\frac{1}{\beta})^{T}\left[b_0u'(c^{\text{FB}}(g_0))-\frac{(1+\beta)\bar{b}T_2(\epsilon)+\beta T_1(\epsilon)}{1-\beta}\right]+\frac{(1+\beta)\bar{b}T_2(\epsilon)+\beta T_1(\epsilon)}{1-\beta}\right\}.
\end{equation}
The LHS of \eqref{eq:incomplete_upper} is bounded above by $\bar{b}$, whereas the RHS diverges to infinity as $T\to\infty$. Hence there exists some $\bar{T}$ depending on $\epsilon, b_0, \beta, g_0$ but not on $\underline{c}$, s.t. $T\le \bar{T}$. For some $t\le \bar{T}+2$ and some $g^t\in\mathcal{G}^t$, we have $|\ell_t(g^t)-\ell^{\text{FB}}(g_t)|>\epsilon$, which implies a welfare loss of at least $T_3(\epsilon)\beta^{\bar{T}+2}\pi_{\min}^{\bar{T}+2}$, where
\[
T_3(\epsilon)=\min_{g\in\mathcal{G}}\inf_{|\ell-\ell^{\text{FB}}(g)|>\epsilon}\left[u(\ell^{\text{FB}}(g)-g)+v(\ell^{\text{FB}}(g))-u(\ell-g)-v(\ell)\right]>0,
\]
and this welfare loss does not depend on $\underline{c}$. Therefore, when $\underline{c}\rightarrow\zeta$, the total welfare of the deterministic solution does not tend to $V_{F}^{\text{FB}}(g_0)$.

\section{Optimal lotteries when markets are incomplete}\label{sec:im}
When the government only has access to a one-period bond, it can be optimal to use lotteries even with CRRA utility, and purification might be impossible even if debt constraints are implicit and $ \beta $ is close to one. To understand the intuition of why CRRA utility does not preclude lotteries, it is useful to first examine a two-period example with lotteries.

\subsection{Two-period intuition}\label{sec:2pim}
We assume $u(c)=\log c$, $v(\ell)=\sqrt{1-\ell}$, and $\beta=1$.
The government can randomize over tax rates in each state $s\in\{1,\cdots, S\}$ in the second period, i.e., the government can randomize between various $\tau_s$ (corresponding to various $\ell_s$) so that it can repay its debt for any $s\in\{1,\cdots, S\}$. The lottery problem can be written as
 \begin{equation}\label{eq:ramsey_reduce_lot}
\begin{aligned}
    &\max_{(\pi_s)_{s=0}^{S}, b\le \bar{b}}u(\ell_0-g_0)+v(\ell_0)+\beta\sum_{s=1}^{S}  p_s \int_{\ell_s:b\le r(\ell_s;g_s)} (u(\ell_s-g_s)+v(\ell_s))\pi_s(d\ell_s),\\
    \text{s.t. } 
    &b_{-1}\le r(\ell_0;g_0)+(\ell_0-g_0)(\underbrace{\beta \sum_{s=1}^{S}p_s \int_{\ell_s:b\le r(\ell_s;g_s)} \frac{b}{\ell_s-g_s}\pi_s(d\ell_s)}_{w'}),
\end{aligned}   
\end{equation}
 where
 \begin{equation}\label{eq:def_h_and_f}
    r(\ell;g):=(\ell-g) \left( 1-\frac{\ell}{2\sqrt{1-\ell}} \right).
\end{equation}
    The function $r(\ell,g)$ can be interpreted as the net surplus of a government that sets the tax rate so as to induce a labor supply of $ \ell $ and spends $g$. 
 The key to understanding why randomized taxes might improve welfare in this example is to note that in the second period there are generally two possible tax rates that yield the same government surplus.
\begin{figure}[tbh]
\centering
\includegraphics[width=0.4\textwidth]{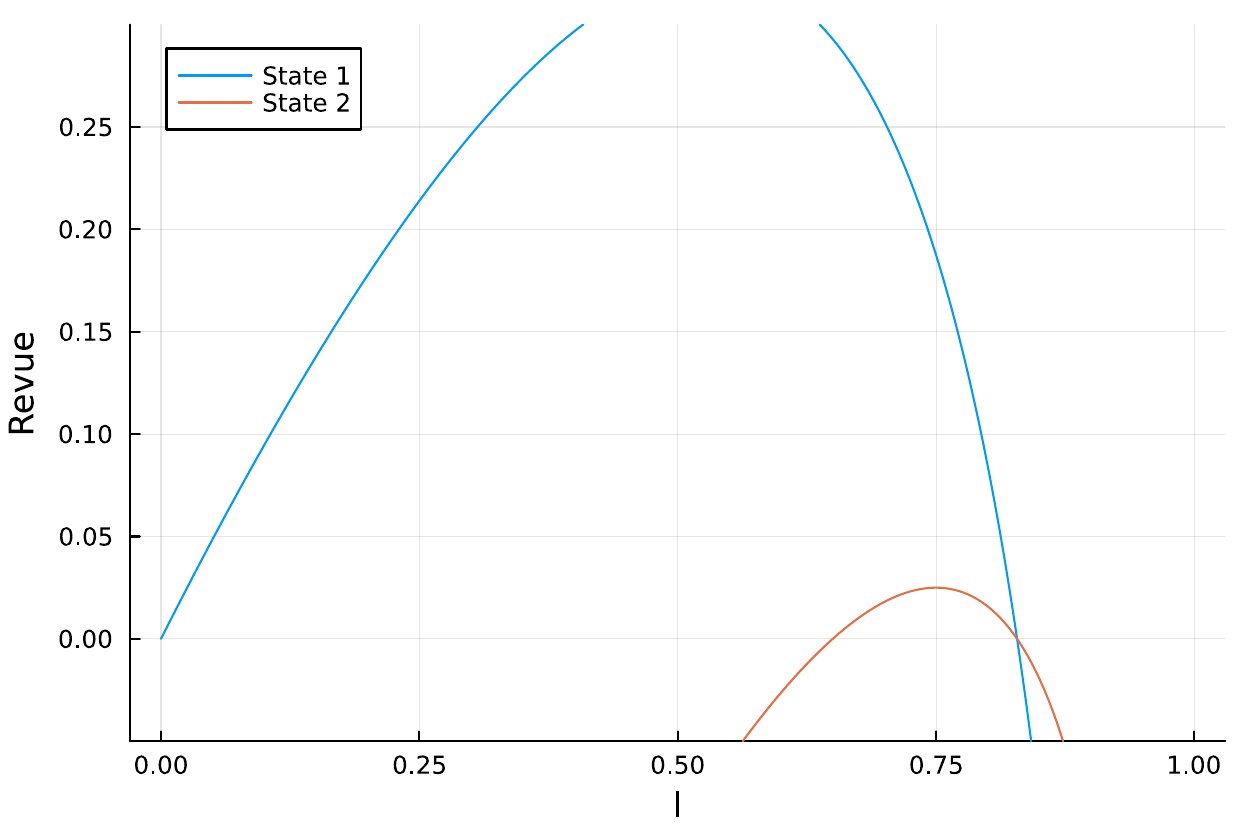}
\includegraphics[width=0.4\textwidth]{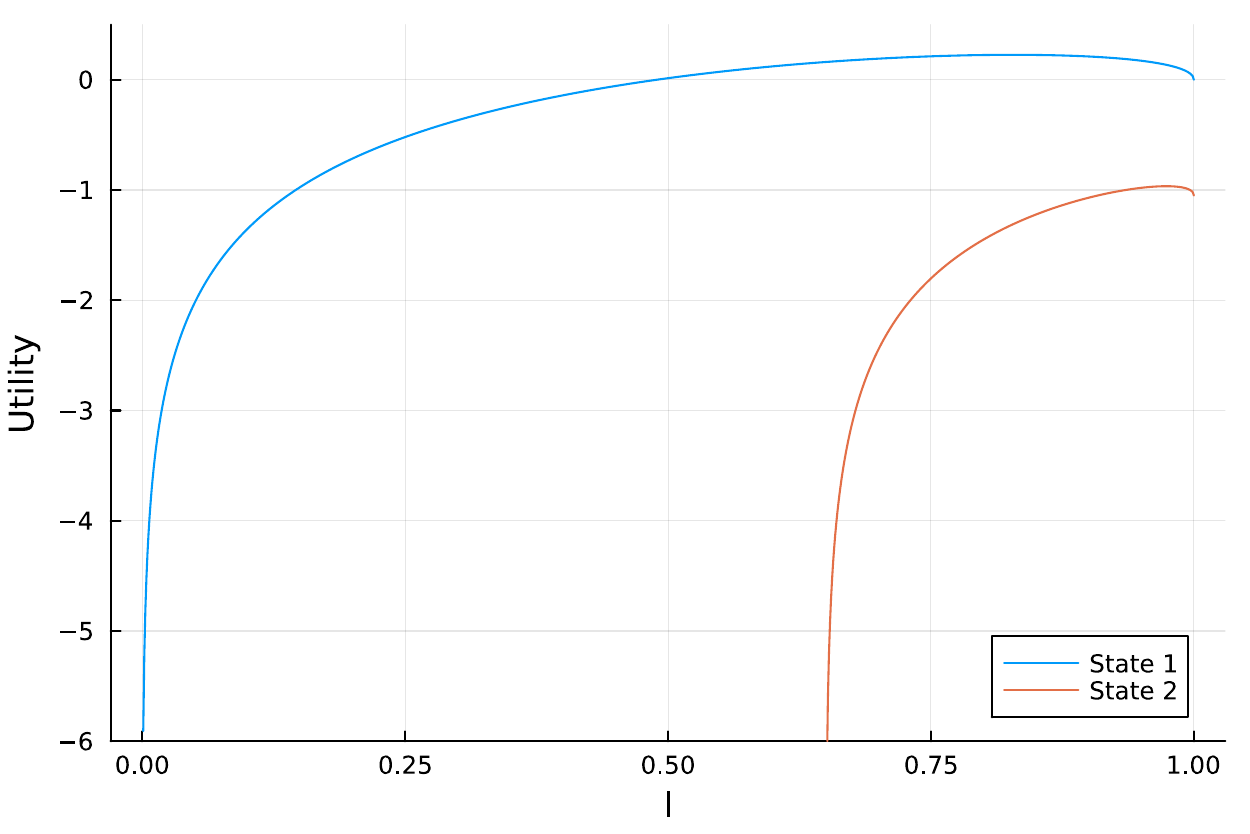}
\caption{\label{fig:fig1} The left panel shows the functions $ r(\ell;g(s)) $ for $ s=1,2 $. The right panel shows utility as a function of $\ell$ .}
\end{figure}
 
To illustrate this, we consider a concrete example and assume $p_L=0.9$, as well as
 $ 0 = g_0=g_L<g_H = 0.65,\,b_{-1}=0.45$. Figure \ref{fig:fig1} shows how much revenue the government can generate in each state $s$ and the associated utility in that state.
Note first that the amount of debt is limited by the government's ability to pay back in state 2. Debt levels above approximately 0.025 are not feasible. Second, since $b_{-1}>\max_{\ell_0} r(\ell_0;g_0)$, the government cannot repay its entire debt in the first period, i.e. only $ b>0 $ is feasible. Third, for a given level of feasible debt, the government can collect the necessary revenue either with a very high tax, which implies low $ \ell $ and low welfare, or with a relatively low tax, which implies high $ \ell $ and high welfare.
The key to why the government might find it optimal to sometimes set very high taxes lies in the fact that it can only generate a large revenue from selling bonds at $ t=0 $ if the interest rate implied by the household's Euler equation \eqref{eq:ramsey_agenteuler} is sufficiently low. Since borrowing is constrained by the large government expenditure in state 2, $ g_2=0.65$, the government might find it optimal to announce a lottery for state 1 that imposes a very high tax with some moderate probability, leading to very low consumption, a high demand for savings, and very low interest rates in period zero. 

The lottery problem \eqref{eq:ramsey_reduce_lot} is solved numerically using the Lagrangian iteration method proposed in \cite{shen2025lagrangian}. For this simple example, the optimal deterministic solution is obtained via grid search, and the resulting solutions are presented in Table \ref{table:comparison}.\footnote{Values reported as 0 indicate absolute values less than $1\times10^{-3}$.}
\begin{table}[htbp]
\centering
\begin{tabular}{lcccccccc}
\toprule
 & $b$ & $\ell_0$ & $\ell_1$ & $\ell_2$ & $T_0$ & $T_1$ & $T_2$ & Welfare \\
\midrule
\multirow{2}{*}{Lottery} 
   & \multirow{2}{*}{0.024} 
   & \multirow{2}{*}{0.625} 
   & 0.024, 21.36\% & \multirow{2}{*}{0.767} 
   & \multirow{2}{*}{0} 
   & 0 & \multirow{2}{*}{0} 
   & \multirow{2}{*}{-0.391} \\
   & & & 0.821, 78.64\% & & & 0 & & \\
\midrule
Deterministic & 0.024 & 0.605 & 0.108 & 0.743 & 0 & 0.078 & 0 & -1.213 \\
\bottomrule
\end{tabular}
\caption{Lottery solution vs. Deterministic solution}\label{table:comparison}
\end{table}
In this example, the government's ability to raise revenue from debt in period zero turns out to improve significantly if lotteries are possible. In the deterministic case, the government must impose a high tax in state \(s=1\) to ensure that the marginal value of debt repayment, captured by \(bu'(c_1)=\frac{b}{\ell_1-g_1}\), is sufficiently large to roll over the initial debt \(b_{-1}\). This results in excess tax revenue \(T_1=0.078\) in that state, which is then transferred back to the agents. Because agents are prudent, \(u'(c_1)\) is convex in consumption \(c_1\). As a result, randomizing the labor tax in state \(s=1\) (and hence the consumption $c_1$ in state \(s=1\)) between a lower rate and a higher rate raises the expected marginal value \(\mathbb{E}bu'(c_1)\), making it easier for the government to roll over its initial debt. On the other hand, agents are also risk-averse, making the period utility \(u(\ell_1-g_1)+v(\ell_1)\) concave in \(\ell_1\). Randomization in the labor tax in $s=1$ therefore reduces the agent's expected utility. When the benefit of randomization (driven by prudence) outweighs its cost (driven by risk aversion), the government can use lotteries to improve social welfare. In the lottery solution, the government randomizes the labor tax in state \(s=1\) and reduces the tax burden in states \(s=0\) and \(s=2\). This leads to higher labor supply in those states and a substantial welfare gain compared to the deterministic solution.

\subsection{An example for the full problem}
We now extend the intuition to an infinite-horizon example. As before,
we assume that $u(c_t)=\log c_t$ and $v(\ell_t)=\sqrt{1-\ell_t}$, and consider the following deterministic problem
\begin{align}  \label{eq:ramsey_obj_infty_twostate}\max_{(1-\epsilon\ge \ell_t\ge g_t+\epsilon,\,b_t)_{t=0}^{\infty}} & \mathbb{E}_0\sum_{t=0}^{\infty}\beta^t\left[\log(\ell_t-g_t)+\sqrt{1-\ell_t}\right], \\ 
\label{eq:ramsey_accumulate_infty_twostate}\mbox{ s.t. } &\frac{b_{t-1}}{\ell_t-g_t}\le \underbrace{1-\frac{\ell_t}{2\sqrt{1-\ell_t}}}_{h(\ell_t)}+\beta\mathbb{E}_t\frac{b_{t}}{\ell_{t+1}-g_{t+1}},\\
& \sup_t |b_t| < \infty.
\end{align}

It is straightforward to verify that $h$ is a decreasing concave function and admits a unique zero $g^*=2\sqrt{2}-2$. In this example, there are two states $g^1=0$, and $g^2=g^*-\sigma$. The transition probability matrix is:
\[P_{\delta}= (p(g'=g^j|g=g^i))_{i,j}=
\begin{pmatrix}
1-\delta & \delta\\
\delta & 1-\delta 
\end{pmatrix}.
\]
The parameters $\sigma$ and $\delta$ are assumed to satisfy the following assumption.
\begin{assumption}\label{twostate_example_ass1}
   We assume that there exists $k>0$, s.t. $\sigma=e^{-k}$, $\epsilon=e^{-3k}$, $\delta=e^{-5k}$. 
\end{assumption}
The assumption ensures that, when $k$ is sufficiently large, $\delta  \ll \epsilon \ll \sigma \ll \frac{1}{k}$.  

For any $b_{t-1}\ge 0$, we define
\begin{equation}\label{equ:defffunc}
    f(\ell_t;b_{t-1},g_t)=\frac{b_{t-1}}{\ell_t-g_t}-h(\ell_t)=\frac{b_{t-1}}{\ell_t-g_t}+\frac{\ell_t}{2\sqrt{1-\ell_t}}-1.
\end{equation}
and rewrite the forward-looking constraint \eqref{eq:ramsey_accumulate_infty_twostate}  as
\begin{equation}\label{eq:ramsey_accumulate_infty_fform_twostate}
f(\ell_t;b_{t-1},g_t)\le \beta\mathbb{E}_t\frac{b_t}{\ell_{t+1}-g_{t+1}}.
\end{equation}
The next lemma shows that $f(\ell_t;b_{t-1},g_t)$ is a U-shaped function in $\ell_t$.
\begin{lemma}\label{lem:infinite_Ramsey_convex_twostate_text}
    For any $g_t$ and $b_{t-1}\ge 0$, $f(\ell_t;b_{t-1},g_t)$ defined in \eqref{equ:defffunc}
    is a convex function in $\ell_t$ and has a unique minimal point $g_t+\epsilon\le \ell_t^*(b_{t-1},g_t)\le 1-\epsilon$. Furthermore, we define 
    \begin{equation}\label{eq:deffstarfunc}
    f^*(b_{t-1},g_t)=f(\ell_t^*(b_{t-1},g_t);b_{t-1},g_t),
    \end{equation}
    $f^*(b_{t-1},g_t)\ge \frac{g_t+\epsilon}{2\sqrt{1-(g_t+\epsilon)}}-1>-1$ is increasing w.r.t. $b_{t-1}$.
\end{lemma}

When $b_{t-1}$ increases, $f(\ell_t;b_{t-1},g_t)$ increases. The next proposition shows that, when $g_0=g^1$ and $b_{-1}$ is greater than some threshold, the U-shaped function $f(\ell_0;b_{-1},g_0)$ is always greater than 0. Inequality \eqref{eq:ramsey_accumulate_infty_fform_twostate} then implies that the government must hold  positive debt for the next period.
\begin{proposition}\label{prop:prop_bn1_twostate}
    Given $g_0=g^1$, there exist $K_1,\,C_1>0$, s.t. when $k>K_1$, 
    \begin{enumerate}
        \item there exists $b_{-1}>0$ s.t. $f^*(b_{-1},g_0)=\frac{1}{k}$;
        \item $\ell_0^*(b_{-1},g_0)> C_1$.
    \end{enumerate}
\end{proposition}

The initial states of this example are assumed to satisfy the following assumption.
\begin{assumption}\label{twostate_example_ass2} The initial conditions satisfy
    \begin{itemize}
    \item $g_0=g^1$;
     \item $b_{-1}>0$ such that $f^*(b_{-1},g_0)=\frac{1}{k}$, for $k>K_1$.
\end{itemize}
\end{assumption}

The positive probability of a long-lasting bad state $g^2$ severely limits the government's feasible debt $b_t$ for $t\ge 1$, as the following proposition shows.
\begin{proposition}\label{prop:small_debt}
    Let $0<\beta<1$. There exist $K_2>0$ and $C_2>0$, s.t. when $k>K_2$, the feasible debt must satisfy $b_t\le C_2\sigma$ for all $t$. 
\end{proposition}
The proposition implies that the implicit debt constraint can be replaced with an explicit constraint as constructed.

Given the assumptions, the government must hold positive debt in the second period to roll over the initial debt, while the positive probability of the bad state limits the amount it can hold. The key question is at what cost the government can raise the bond price in the first period. We consider a lottery solution in which the government randomizes taxes in the second period at the low-expenditure state. This raises the first-period bond price due to the precautionary saving motive of agents. In the deterministic problem, by contrast, rolling over the debt requires the government to choose a high tax rate (and hence low consumption) deterministically in the second period. The next two propositions compare the welfare of the lottery described above with that of the deterministic solution, showing that this lottery yields higher.
\begin{proposition}[Lower bound of the optimal utility with lotteries]\label{prop_lb_twostate}
    There exist $C_{\text{lot}}\in\mathbb{R}$ and $K_{\text{lot}}>0$, s.t. the optimal utility with lottery is not less than $C_{\text{lot}}$ when $k>K_{\text{lot}}$.
\end{proposition}

\begin{proposition}[Upper bound of the optimal utility without lotteries]\label{prop_ub_twostate}
    There exist $C_{\det}\in\mathbb{R}$ and $K_{\det}>0$, s.t. the optimal utility without lottery is not greater than $\beta(1-\delta)\log \sigma+\beta(1-\delta)\log  k+C_{\det}$ when $k>K_{\det}$.
\end{proposition}
Since $\log\sigma=-k$, the upper bound in Proposition~\ref{prop_ub_twostate} tends to $-\infty$ as $k\to\infty$, while the lower bound in Proposition~\ref{prop_lb_twostate} does not depend on $k$. Hence, for $k$ sufficiently large, the lottery solution strictly dominates every deterministic solution.

\subsection{Unbounded relative risk aversion}
The argument showing that, with unbounded relative risk aversion, a lottery solution can approximate the first best extends directly to the incomplete markets case. For any $g\in\mathcal{G}$, the first-best labor is
$$
\ell^{\text{FB}}(g)=\arg\max_{\ell}u(\ell-g)+v(\ell).
$$
We define $V^{\text{FB}}(g)=\mathbb{E}_{g_0=g}\sum_{t=0}^{\infty}\beta^t\left[u(\ell^{\text{FB}}(g_t)-g_t)+v(\ell^{\text{FB}}(g_t))\right]$. We consider the functional equation \eqref{eq:fe} and let $ W_{\underline{c}}(b,\mu,g)  $ denote its solution as a function of the lower bound on consumption, $ \underline{c} $.
\begin{proposition} \label{prop:IM_lott_fb}
  Under Assumption \ref{ass:regu}, Assumption \ref{ass:v}(\ref{enum_assv_1}) and Assumption \ref{ass:lot_improve}, the lottery problem achieves the first best when $\underline{c}\rightarrow \zeta+$, i.e.
$$\lim_{\underline{c}\rightarrow \zeta} W_{\underline{c}}(b,0,g) = V^{\text{FB}}(g) $$ for each $g\in\mathcal{G}$.
    \end{proposition}

\section{Numerical examples}\label{sec:num}

To solve the functional equation \eqref{eq:fe}
numerically,
we discretize savings and labor supply, but leave the multiplier to be continuous.
Suppose that the set of possible labor choices contains $ n_{\ell} $ elements, assumed to lie strictly between zero and one, and the set of possible bond choices contains $ n_b $ elements. By feasibility, there are also $ n_{\ell} $ possible consumption choices, all of which are assumed to be strictly positive. We describe the algorithm in  Appendix \ref{app:algo}. A key insight is that with discrete savings and labor choices, the function $W(\cdot)$ is piecewise linear in $ \mu $ and therefore can be stored exactly.

\subsection{Multiplicity}\label{sec:reach}
Since we discretize savings and labor supply, it is often optimal to randomize between adjacent values in the grid. These ``trivial lotteries'' are not very interesting.
We therefore focus on lotteries that randomize between tax rates and labor supplies that ``differ substantially.''
In our calibration, the neighboring points of the labor grid are
$0.0064$ apart or less. We require a non-trivial lottery to randomize over points in the grid that are significantly further apart and take the threshold to be 0.1.

To illustrate the mechanism, we consider the following numerical example in detail.
\begin{table}[H]
\centering
\begin{tabular}{llll}
\toprule
Preferences & $u(c)=\log c$ & $v(\ell)=\psi(1-\ell)^{\eta}$ & $\psi=1$, $\eta=0.2$\\
Discounting & $\beta=0.9$ & & \\
Spending    & $g\in\{0,\,0.5\}$ & $\pi_{ij}=1/2$ & \\
Bounds      & $b\in[-1,\,0.5]$ & $c\ge0.02$ & $\ell<1$\\
Grids       & $n_\ell=153$ & $n_b=305$ & \\
\bottomrule
\end{tabular}
\caption{Example 1}\label{tab:tab2}
\end{table}

Figure~\ref{fig:kink} shows the set of genuine ties in Example 1, i.e., the states at which it is optimal to randomize between substantially different allocations. For a wide range of positive initial debt levels, there exists exactly one value of the co-state where randomization is optimal. Panel (b) shows that at these ties randomization is substantial.

\begin{figure}[htp]
\centering
\includegraphics[width=10cm]{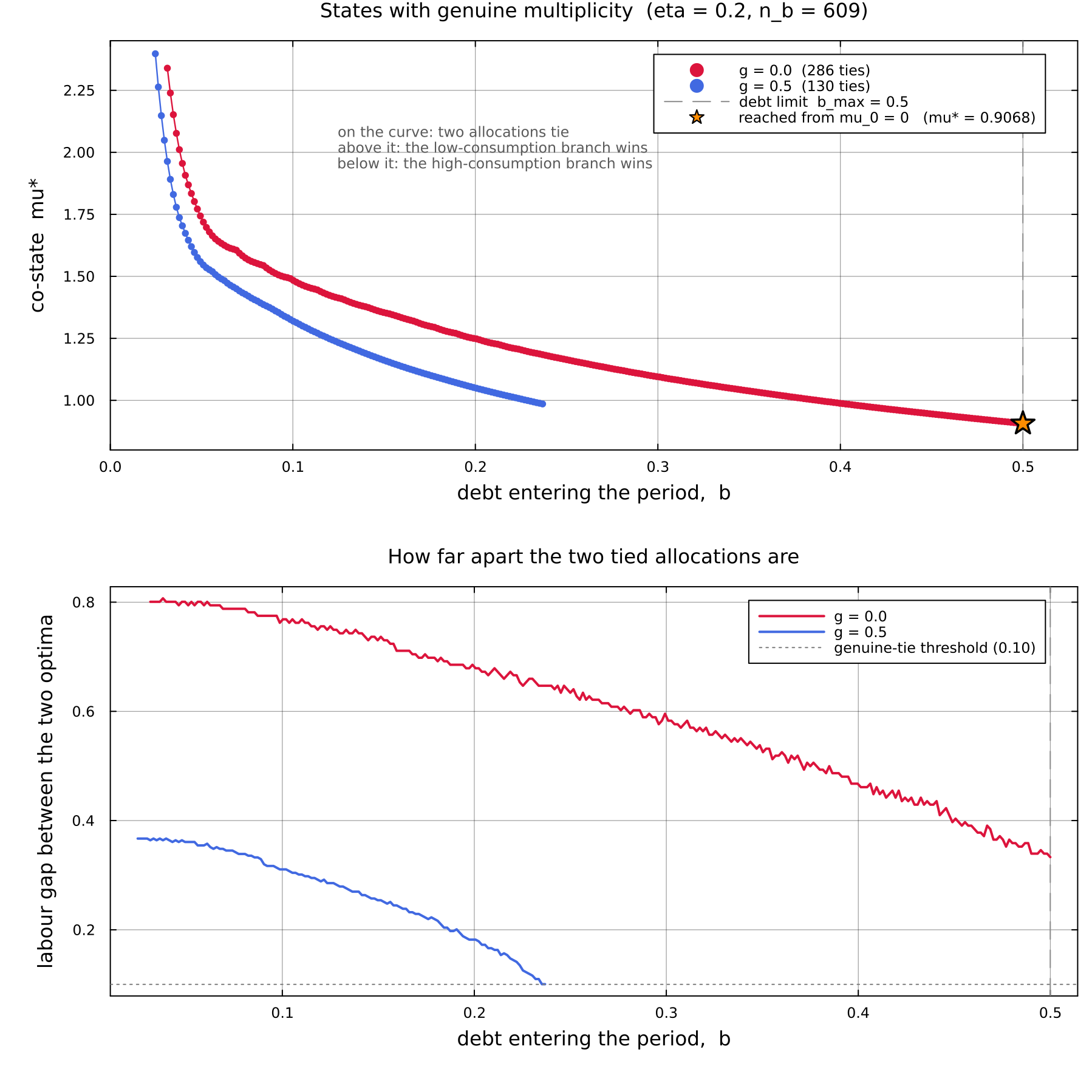}
\caption{Genuine ties:
\textbf{(a)} Set of states with multiplicity
\textbf{(b)} Difference of the two optimal labor choices}
\label{fig:kink}
\end{figure}

\subsection{Reachability}\label{sec:worked}
While there typically exist regions in the state space where lotteries are the optimal solution to the sup-inf problem, these are always ``small.''{} Whether they are reached from some initial condition (with initial $ \mu=0$) is a question we can only answer computationally.

In Example 1, lotteries occur along the optimal solution if we start with a sufficiently high debt level. As explained in the appendix and in \cite{shen2025recursive}, to describe optimal lotteries, it is useful to keep track of the (marginal) utility promises. With
$n_\ell=153$, $n_b=305$, starting at the debt limit with $ b_{-1}=0.5$, $\mu_0=0$ and zero promise in the high-spending state, consider the following sequence of shocks.
\begin{itemize}
\item $t=0$, high $g$: the promise is slack, no lottery. Consumption is $0.362$, the planner stays at
$\bp=0.5$, and hands on the promise $1.6887$ and $\mu=0.9068$.
\item $t=1$, low $g$: the state is $b=0.5$, $g=0$ with the inherited promise
$1.1042$, which lies strictly inside the sub-differential at the kink $\mu^{*}=0.9068$. The
planner's optimal policy is a lottery with the following two ``branches'':
\end{itemize}
\begin{center}
\small
\begin{tabular}{lccccc}
\toprule
branch & probability & $c$ & $\bp$ & $\gamma^{*}$ & promise handed on ($g'=0$, $g'=0.5$)\\
\midrule
high & $0.549$ & $0.6605$ & $0.0428$ & $0$ & $0.051$, \ $0.109$\\
low & $0.451$ & $0.3275$ & $0.3355$ & $0.323$ & $0.415$, \ $0.956$\\
\bottomrule
\end{tabular}
\end{center}
 The two branches represent very different policies: the high branch consumes twice as much and clears almost all the debt; the low
branch consumes $0.33$ and keeps $\bp=0.34$. As explained in Section \ref{sec:2pim}, the low-consumption branch is chosen with fairly high probability to push down the interest rate in the preceding period, making it cheaper to roll over the debt and allowing higher consumption in times where government spending needs to be high.
After the lottery,
with $b=0.3355 $, the high-spending state has to occur four times in a row to push the debt up again to the limit, and $ \mu $ to 0.9068, from where the cycle can repeat itself. In the low-spending state, there are never lotteries for $b<0.5$, and the government saves each period. When starting at zero debt, nine consecutive high-spending shocks are needed for government debt to reach the upper limit. The overall probability of reaching the upper debt limit from zero debt within the first 50 periods is approximately 0.005. Unless initial debt is high, lotteries therefore appear to have a small quantitative impact on welfare; they become important only at high debt levels.

We list many solved examples in Appendix \ref{app:examples}.
We consider the case where
$u(c)=\log c$, $v(\ell)=\psi(1-\ell)^{\eta}$ with $\eta\in(0,1)$,
$c=\ell-g\in[c_{\min},c_{\max}]$ and solve examples with various values for $ \psi, \eta$ and the shock process $(g_t)$.
All of the hundreds of numerically solved examples have genuine ties. That is to say, there are always regions in the state space where it becomes optimal to randomize labor taxes. However, for most of the solved examples, these regions are not reachable.

\subsection{Shifted log}
We now consider a utility function with a potentially unbounded coefficient of relative risk aversion,
but choose $ \underline{c} > \zeta $ to bound $-cu''(c)/u'(c)$.
The first-best cannot be attained, but lotteries still play an important role for the optimal solution.

Assume that preferences are
\[
  u(c)+v(\ell)=\log(c-\zeta)+\psi\,(1-\ell)^{\eta},\qquad
  \zeta=0.05,\quad \psi=1,\quad \eta=\tfrac12,
\]
with $\beta=0.9$, $c=\ell-g$, spending $g\in\{0,0.5\}$ i.i.d.\ and fair, and debt in
$[-1,\,0.5]$ on the usual grid. The consumption floor is $c\ge\underline{c}=0.06$.
In this example, starting with zero debt, 
 it takes eight consecutive high-spending shocks for the government to randomize taxes. Seven bad shocks carry the debt to the
limit; at the eighth bad shock, the government randomizes, even though this is a high-expenditure shock. In the table, $\varphi$ denotes the promise inherited from the previous period.
\begin{center}
\small
\begin{tabular}{cccccccc}
\toprule
$t$ & $b$ & $\varphi$ & $c$ & $\ell$ & $\bp$ & $\gamma^{*}$ & hands on ($g'=0$, $g'=0.5$)\\
\midrule
$0$ & $0.0000$ & $0.0000$ & $0.3792$ & $0.8792$ & $0.0559$ & $0.0000$ & $0.0794$, \ $0.1729$\\
$1$ & $0.0559$ & $0.1729$ & $0.3735$ & $0.8735$ & $0.1184$ & $0.2926$ & $0.1696$, \ $0.3771$\\
$2$ & $0.1184$ & $0.3771$ & $0.3620$ & $0.8620$ & $0.1809$ & $0.1855$ & $0.2630$, \ $0.5799$\\
$3$ & $0.1809$ & $0.5799$ & $0.3620$ & $0.8620$ & $0.2664$ & $0.3565$ & $0.3992$, \ $0.8889$\\
$4$ & $0.2664$ & $0.8889$ & $0.3505$ & $0.8505$ & $0.3651$ & $0.4151$ & $0.5680$, \ $1.2542$\\
$5$ & $0.3651$ & $1.2542$ & $0.3418$ & $0.8418$ & $0.4901$ & $0.4876$ & $0.8007$, \ $1.7286$\\
$6$ & $0.4901$ & $1.7286$ & $0.3332$ & $0.8332$ & $0.5000$ & $0.6059$ & $0.8882$, \ $2.6100$\\
\midrule
\multicolumn{8}{l}{$t=7$: \ $b=0.5000$, $\varphi=2.6100$ --- a  lottery at
$\mu^{*}=0.6059$, consumption spread $0.2531$:}\\
$7$ & \multicolumn{2}{c}{w.p.\ $0.9853$} & $0.3131$ & $0.8131$ & $0.5000$ & $0.6059$
    & $0.8882$, \ $2.7804$\\
$7$ & \multicolumn{2}{c}{w.p.\ $0.0147$} & $0.0600$ & $0.5600$ & $0.5000$ & $0.6147$
    & $48.7158$, \ $50.0000$\\
\bottomrule
\end{tabular}
\end{center}

 The date-$7$ lottery has two branches, each followed by
two possible shocks; the resulting in four date-$8$ states are as follows.
\begin{center}
\small
\begin{tabular}{llll}
\toprule
branch at $t=7$ & $g_8$ & state at $t=8$ & what the planner does\\
\midrule
$c=0.3131$ & $0$   & $\varphi=0.888192$ & \emph{no} lottery: $c=0.6129$, $\bp=0.1382$; leaves the limit\\
$c=0.3131$ & $0.5$ & $\varphi=2.7804$   & lottery, $\mu^{*}=0.6059$: $0.9817$ stays, $0.0183$ to the ceiling\\
$c=0.0600$ & $0$   & $\varphi=48.7158$  & lottery, $\mu^{*}=0.6147$: $0.9739$ floor, $0.0261$ exit\\
$c=0.0600$ & $0.5$ & $\varphi=50$       & \emph{no} lottery: $c=0.0600$, $\bp=0.5$; the ceiling fixed point\\
\bottomrule
\end{tabular}
\end{center}
The first row is the interesting one, because $b=0.5$, $g=0$ does carry a genuine tie and yet
nothing is randomized. The promise handed on, $0.888192$, is exactly the lower slope
$b\,u'(0.6129)=0.5/0.5629$ of that kink, so it sits at the corner of the sub-differential and
the weight on the floor allocation is zero. The same state with an interior promise---row
three, $\varphi=48.7158$---mixes the same two allocations, $c=0.6129$ and $c=0.0600$, with
weights $0.0261$ and $0.9739$. Thus a single kink supports a family of lotteries indexed by the promise, and at
the endpoints of that family the lottery degenerates.

\section{Conclusion}\label{sec:concl}
We examine two simple models of optimal government policy. Despite their simplicity, the models are difficult to analyze because they reduce to non-convex optimization problems. We use the techniques developed in \cite{shen2025lagrangian} and in \cite{shen2025recursive} to examine the role of public randomization for the optimal solutions.
Depending on the market structure and preferences, lotteries can drastically improve welfare. However, full commitment is crucial for this result, as high-tax lottery outcomes are always very costly in terms of welfare ex post.
 \cite{jiang2025implementing} examine conditions on utility under which the Ramsey solution of the complete markets model is time consistent.
 Their results are consistent with ours, since they focus on CRRA utility.

\onehalfspacing
 \bibliographystyle{ecta}
\bibliography{refe}

\renewcommand{\appendixpagename}{Appendix}
\begin{appendices}
\section{Proofs} \label{app:proofs}
\subsubsection*{Proof of Proposition \ref{prop:smallM}}
    When $M\le 2$, for any $g\in \mathcal{G}$ and $c=\ell-g$, we have 
    $$
    G''(c,\ell)=\underbrace{u'''(c)c+2u''(c)}_{\le 0\text{ by Assumption \ref{ass:v}}}+\underbrace{v'''(\ell)\ell+2v''(\ell)}_{< 0\text{ by Assumption \ref{ass:v}}}< 0.
    $$
    Hence $G(c,\ell)$ is concave. According to Sion's minimax theorem, we have
    $$
    \inf_{\lambda\ge 0}\max_{(\ell_t)_{t=1}^{\infty}\in L} \mathcal{L}(\ell_0,(\ell_t)_{t=1}^{\infty};\lambda)=\max_{(\ell_t)_{t=1}^{\infty}\in L}\inf_{\lambda\ge 0}\mathcal{L}(\ell_0,(\ell_t)_{t=1}^{\infty};\lambda),
    $$
    for any $\ell_0$, and there is no lottery improvement.
\hfill $\Box$
\subsubsection*{Proof of Lemma \ref{lem:M>2}}
By direct computation,
$$
\begin{aligned}
&\left[F(c_t,\ell_t)+\lambda G(c_t,\ell_t)\right]''\\=&\left[(1+2\lambda)u''(c_t)+\lambda u'''(c_t)c_t\right]+\left[(1+2\lambda)v''(\ell_t)+\lambda v'''(\ell_t)\ell_t\right]\\
=&\lambda (-u''(c_t))\left[-\frac{u'''(c_t)c_t}{u''(c_t)}-2-\frac{1}{\lambda}\right]+\underbrace{\left[(1+2\lambda)v''(\ell_t)+\lambda v'''(\ell_t)\ell_t\right]}_{< 0\text{ by Assumption \ref{ass:v}}}.
\end{aligned}
$$
When $\lambda\le \frac{1}{M-2}$, we have
$$
-\frac{u'''(c_t)c_t}{u''(c_t)}-2-\frac{1}{\lambda}\le -\frac{u'''(c_t)c_t}{u''(c_t)}-M\le 0.
$$
Hence $F(c_t,\ell_t)+\lambda G(c_t,\ell_t)$ is concave.
\hfill $\Box$
\subsubsection*{Proof of Proposition \ref{prop:largeM}}
When $\lambda>\lambda^*$, we have
$$
\begin{aligned}
&\mathcal{L}(\ell_0^*,(\ell_t)_{t=1}^{\infty};\lambda)\\
=&F(c_0^*,\ell_0^*)+\mathbb{E}_{g_0}\sum_{t=1}^{\infty}\beta^tF(c_t,\ell_t)+\lambda\left(G(c_0^*,\ell_0^*)+\mathbb{E}_{g_0}\sum_{t=1}^{\infty}\beta^tG(c_t,\ell_t)- b_0u'(c_0^*)\right).
\end{aligned}
$$
We take $(\ell_t)_{t=1}^{\infty}=(\ell_t^{\lambda^*})_{t=1}^{\infty}$, and have
$$
\begin{aligned}
&\max_{(\ell_t)_{t=1}^{\infty}\in L}\mathcal{L}(\ell_0^*,(\ell_t)_{t=1}^{\infty};\lambda)\ge \mathcal{L}(\ell_0^*,(\ell_t^{\lambda^*})_{t=1}^{\infty};\lambda)\\
=& \max_{(\ell_t)_{t=1}^{\infty}\in L}\mathcal{L}(\ell_0^*,(\ell_t)_{t=1}^{\infty};\lambda^*)+(\lambda-\lambda^*)\left(G(c_0^*,\ell_0^*)+\mathbb{E}_{g_0}\sum_{t=1}^{\infty}\beta^tG(c_t^{\lambda^*},\ell_t^{\lambda^*})- b_0u'(c_0^*)\right)\\
>&\max_{(\ell_t)_{t=1}^{\infty}\in L}\mathcal{L}(\ell_0^*,(\ell_t)_{t=1}^{\infty};\lambda^*).
\end{aligned}
$$
Therefore, we have
$$
\min_{\lambda\ge 0}\max_{(\ell_t)_{t=1}^{\infty}\in L}\mathcal{L}(\ell_0^*,(\ell_t)_{t=1}^{\infty};\lambda)=\min_{0\le \lambda\le \lambda^*}\max_{(\ell_t)_{t=1}^{\infty}\in L}\mathcal{L}(\ell_0^*,(\ell_t)_{t=1}^{\infty};\lambda).
$$
According to Lemma \ref{lem:M>2}, when $\lambda\le\lambda^*$, $\mathcal{L}$ is concave on $(\ell_t)_{t=1}^{\infty}$, and hence the maximizer is unique and there is no lottery improvement (see Corollary 2.1 in \cite{shen2025lagrangian}).
\hfill $ \Box$
\subsubsection*{Proof of Proposition \ref{prop:lot_improve}}
        Taking $\lambda=0$, we know that the LHS of \eqref{eq:goal_lot_improvement_primal} is not greater than the RHS. 
        It suffices to show that
        \begin{equation}\label{eq:goal_lot_improvement}
        \lim_{\underline{c}\rightarrow\zeta} \max_{\underline{c}+g_0\le \ell_0\le \bar{\ell}}\min_{\lambda\ge 0}\max_{(\ell_t)_{t=1}^{\infty}\in L}\mathcal{L}(\ell_0,(\ell_t)_{t=1}^{\infty};\lambda)\ge \max_{(\zeta+g_t\le \ell_t\le \bar{\ell})_{t=0}^{\infty}}\mathbb{E}_{g_0}\sum_{t=0}^{\infty}\beta^tF(c_t,\ell_t).
        \end{equation}
        The right-hand side is solved by the first-best solution, denoted by $(\ell_t^{\text{FB}})_{t=0}^{\infty}$, which satisfies
$$
u'(\ell_t^{\text{FB}}(g^t)-g_t)+v'(\ell_t^{\text{FB}}(g^t))=0.
$$
        Therefore, we have
        \begin{equation}\label{eq:goal_lot_improvement_step1}
            \begin{aligned}
            &\max_{\underline{c}+g_0\le \ell_0\le \bar{\ell}}\min_{\lambda\ge 0}\max_{(\ell_t)_{t=1}^{\infty}\in L}\mathcal{L}(\ell_0,(\ell_t)_{t=1}^{\infty};\lambda)\\\ge& \min_{\lambda\ge 0}\max_{(\ell_t)_{t=1}^{\infty}\in L}\mathcal{L}(\ell_0^{\text{FB}},(\ell_t)_{t=1}^{\infty};\lambda)\\
            =&\min_{\lambda\ge 0}\max_{(\ell_t)_{t=1}^{\infty}\in L}\lambda \left[G(c_0^{\text{FB}},\ell_0^{\text{FB}})- b_0u'(c_0^{\text{FB}})+\mathbb{E}_{g_0}\sum_{t=1}^{\infty}\beta^tG(c_t,\ell_t)\right]+\left[F(c_0^{\text{FB}},\ell_0^{\text{FB}})+\mathbb{E}_{g_0}\sum_{t=1}^{\infty}\beta^tF(c_t,\ell_t)\right].
            \end{aligned}
        \end{equation}
        For any $\lambda\ge 0$, we have
        $$
        \begin{aligned}
            D(\lambda;\underline{c}):=&\max_{(\ell_t)_{t=1}^{\infty}\in L}\lambda \left[G(c_0^{\text{FB}},\ell_0^{\text{FB}})- b_0u'(c_0^{\text{FB}})+\mathbb{E}_{g_0}\sum_{t=1}^{\infty}\beta^tG(c_t,\ell_t)\right]+\left[F(c_0^{\text{FB}},\ell_0^{\text{FB}})+\mathbb{E}_{g_0}\sum_{t=1}^{\infty}\beta^tF(c_t,\ell_t)\right]\\
            \ge&\lambda \underbrace{\left[G(c_0^{\text{FB}},\ell_0^{\text{FB}})- b_0u'(c_0^{\text{FB}})+\mathbb{E}_{g_0}\sum_{t=1}^{\infty}\beta^tG(\underline{c},\underline{c}+g_t)\right]}_{I}+\underbrace{\left[F(c_0^{\text{FB}},\ell_0^{\text{FB}})+\mathbb{E}_{g_0}\sum_{t=1}^{\infty}\beta^tF(\underline{c},\underline{c}+g_t)\right]}_{J}\\
            =&I\left[\lambda-\frac{J}{I}\right].
        \end{aligned}
        $$

        We now introduce the following technical lemma.
\begin{lemma}\label{lem:tech_prop33}
    Under assumption \eqref{ass:lot_improve}, the utility function for consumption $u$ satisfies
    $$
    \lim_{c\rightarrow\zeta+}u'(c)c=\infty,\quad\lim_{c\rightarrow\zeta+}\frac{u(c)}{u'(c)c}=0.
    $$
\end{lemma}
\begin{proof}
We first show that $\lim_{c\rightarrow\zeta+} u'(c)c=\infty$. When $\zeta>0$, this follows directly from the Inada condition $\lim_{c\rightarrow\zeta+}u'(c)=\infty$. When $\zeta=0$, we choose \(\bar c>0\) such that \(R(c)>2\) for all \(c<\bar c\). Then for any \(0<\underline c<\bar c\),
$$
\log u'(\underline c)-\log u'(\bar c)
=\int_{\underline c}^{\bar c}\frac{R(c)}{c}\,dc
>2\int_{\underline c}^{\bar c}\frac{1}{c}\,dc
=2(\log\bar c-\log\underline c).
$$
Hence
$$
\log\bigl(u'(\underline c)\underline c\bigr)
=\log u'(\underline c)+\log\underline c
>\log u'(\bar c)+2\log\bar c-\log\underline c
\to\infty
\quad(\underline c\to0^+),
$$
implying that \(u'(\underline c)\underline c\to\infty(\underline{c}\rightarrow0+)\).

We then show that $\lim_{c\rightarrow\zeta+}\frac{u(c)}{u'(c)c}=0$. Indeed, by L'Hôpital's rule, we have
$$
\lim_{c\to\zeta^+}\frac{u(c)}{u'(c)c}
=\lim_{c\to\zeta^+}\frac{u'(c)}{u''(c)c+u'(c)}
=\lim_{c\to\zeta^+}\frac{1}{1-R(c)}
=0,
$$
and we finish the proof for Lemma \ref{lem:tech_prop33}.
\end{proof}
According to Lemma \ref{lem:tech_prop33}, we have
        $$
        I\rightarrow\infty,\quad \frac{J}{I}\rightarrow0,\quad \text{as }\underline{c}\rightarrow
\zeta.$$
        Therefore, we have $\lambda(\underline{c})\in \arg\min_{\lambda\ge 0}D(\lambda;\underline{c})\rightarrow 0$ as $\underline{c}\rightarrow
\zeta$.
Hence
\begin{equation}\label{eq:goal_lot_improvement_step2}
    \begin{aligned}
        &\lim_{\underline{c}\rightarrow
\zeta} \min_{\lambda\ge 0}\max_{(\ell_t)_{t=1}^{\infty}\in L}\lambda \left[G(c_0^{\text{FB}},\ell_0^{\text{FB}})- b_0u'(c_0^{\text{FB}})+\mathbb{E}_{g_0}\sum_{t=1}^{\infty}\beta^tG(c_t,\ell_t)\right]+\left[F(c_0^{\text{FB}},\ell_0^{\text{FB}})+\mathbb{E}_{g_0}\sum_{t=1}^{\infty}\beta^tF(c_t,\ell_t)\right]\\
\ge &\lim_{\underline{c}\rightarrow
\zeta}\lambda(\underline{c}) \left[\mathbb{E}_{g_0}\sum_{t=0}^{\infty}\beta^tG(c^{\text{FB}}_t,\ell^{\text{FB}}_t)- b_0u'(c_0^{\text{FB}})\right]+\left[\mathbb{E}_{g_0}\sum_{t=0}^{\infty}\beta^tF(c^{\text{FB}}_t,\ell^{\text{FB}}_t)\right]=\mathbb{E}_{g_0}\sum_{t=0}^{\infty}\beta^tF(c^{\text{FB}}_t,\ell^{\text{FB}}_t).
    \end{aligned}
\end{equation}
We combine \eqref{eq:goal_lot_improvement_step1} and \eqref{eq:goal_lot_improvement_step2} to conclude that \eqref{eq:goal_lot_improvement} holds. \hfill $ \Box $
 
\subsubsection*{Proof of Lemma \ref{lem:finite_support_strong}}

    According to Assumption \ref{ass:utility_regu}, $H_{\lambda}(\ell-g,\ell)$ is not a constant on $[\underline{c}+g,\bar{\ell}]$. Note that $\ell$ is a maximum point of $H_{\lambda}(\ell-g,\ell)$ only if $H_{\lambda}'(\ell-g,\ell)=0$ or $\ell\in\{\underline{c}+g,\bar{\ell}\}$. It suffices to show that there exists $N_g\ge 0$, s.t. for any $\lambda\ge 0$, $H'_{\lambda}(\ell-g,\ell)$ has at most $N_g$ zeros on $(\underline{c}+g,\bar{\ell})$.
    \begin{itemize}
        \item If $G'(\ell-g,\ell)\equiv 0$, it follows that $F'(\ell-g,\ell)\not\equiv 0$. Due to its analyticity, $F'(\ell-g,\ell)$ has a finite number of zeros, denoted as $N^F_g$. And we can take $N_g=N^F_g$.
        \item If $G'(\ell-g,\ell)\not\equiv 0$, then $G'(\ell-g,\ell)$ has finitely many zeros $\underline{c}+g\le \ell_1\le \cdots\le  \ell_{N_g^G} \le \bar{\ell}$. For convenience, we define $\ell_0=\underline{c}+g,\ell_{N_g^G+1}=\bar{\ell}$. In each interval $(\ell_k,\ell_{k+1})$, 
        $$
        \left(\frac{F'(\ell-g,\ell)}{G'(\ell-g,\ell)}\right)'=\frac{F''(\ell-g,\ell)G'(\ell-g,\ell)-G''(\ell-g,\ell)F'(\ell-g,\ell)}{[G'(\ell-g,\ell)]^2}.
        $$
        We define 
        $$
        I_g(\ell)=F''(\ell-g,\ell)G'(\ell-g,\ell)-G''(\ell-g,\ell)F'(\ell-g,\ell),
        $$
        and $I_g(\ell)$ has a finite number of zeros\footnote{$I_g$ is not identically zero because $F'/G'$ is nonconstant by Assumption \ref{ass:utility_regu}.} in $[\underline{c}+g,\bar{\ell}]$, denoted as $N_g^I$. We assume that $I_g$ has $n_k$ zeros in $(\ell_k,\ell_{k+1})$, then $\sum_{k=0}^{N_g^G} n_k\le N_g^I$. By Rolle's theorem, for any $\lambda\ge 0$, there are at most $n_{k}+1$ zeros of 
        $$
        \frac{F'(\ell-g,\ell)}{G'(\ell-g,\ell)}=\lambda
        $$
        in $(\ell_k,\ell_{k+1})$. Therefore, there are at most
        $$
        (N^G_g+2)+\sum_{k=0}^{N_g^G}(n_k+1)\le 2N_g^G+3+N_g^I:=N_g
        $$
        zeros of $H_{\lambda}'(\ell-g,\ell)$. \hfill $ \Box $
    \end{itemize}

 \subsubsection*{Proof of Lemma \ref{lem:purify}}
        \begin{itemize}
            \item \textbf{Step 1.} Since the positive sequence
\[
(\frac{\beta^{t}\pi(g^{t}\mid g_0)}
{\sum_{s=1}^{\infty}\sum_{h^s\in\mathcal{G}^s:\,h_s=g}\beta^s\pi(h^s\mid g_0)})_{t\ge 1,\;g^t\in\mathcal{G}^t\text{ s.t. }g_t=g}
\]
is summable, it can be rearranged in non-increasing order as
\[
q^1\ge q^2\ge \cdots ,
\]
where each $q^i\ge 0$ and $\sum_{i=1}^{\infty}q^i=1$.
\item \textbf{Step 2.} We consider the following procedure. Initialize
\[
\tilde{p}_k^0=0,\qquad k=1,\dots,M_g.
\]
At step \(i\ge 1\) we search for the smallest \(k\in\{1,\dots,M_g\}\) such that
\[
q^i\le p_k-\tilde{p}_k^{i-1}.
\]
If such a \(k\) exists, say \(k_i\), we update
\[
\tilde{p}_{k_i}^{i}=\tilde{p}_{k_i}^{i-1}+q^i,
\]
and set $\tilde{p}_j^{i}=\tilde{p}_j^{i-1}$ for all \(j\ne k_i\).
We claim that this procedure can continue indefinitely, i.e., at every step \(i\ge 1\) such a \(k_i\) exists. We prove this by contradiction. Assume that the procedure proceeds up to some step \(I\ge 1\) and then fails, so that
\[
q^I>p_k-\tilde{p}_k^{I-1}
\]
for every \(k\in\{1,\dots,M_g\}\). On one hand, we have
\begin{equation}\label{eq:estimate_Greater}
M_gq^I>\sum_{k=1}^{M_g}p_k-\sum_{k=1}^{M_g}\tilde{p}_k^{I-1}=1-\sum_{k=1}^{M_g}\tilde{p}_k^{I-1}=1-\sum_{i=1}^{I-1}q^i=\sum_{i=I}^{\infty}q^i.
\end{equation}
On the other hand, suppose that
\[
q^I=\frac{\beta^{t}\pi(g^{t}\mid g_0)}
{\sum_{s=1}^{\infty}\sum_{h^s\in\mathcal{G}^s:\,h_s=g}\beta^s\pi(h^s\mid g_0)}
\]
for some \(t\ge 1\) and \(g^t\in\mathcal{G}^t\) with \(g_t=g\). We write
\[
g^t\prec g^{t+\tau}
\]
to mean that \(g^{t+\tau}\) is a successor of \(g^t\). Then, for every successor \(g^{t+\tau}\) of \(g^t\) where $g_{t+\tau}=g$,  we have
$$
\beta^{t+\tau}\pi(g^{t+\tau}\mid g_0)< \beta^t\pi(g^t\mid g_0),
$$
and hence
\begin{equation}\label{eq:estimate_Smaller}
\begin{aligned}
    \sum_{i=I}^{\infty}q^i\ge&q^I+ \frac{1}
{\sum_{s=1}^{\infty}\sum_{h^s\in\mathcal{G}^s:\,h_s=g}\beta^s\pi(h^s\mid g_0)}\sum_{\tau=1}^{\infty}\sum_{g^{t+\tau}\in\mathcal{G}^{t+\tau}:\,g^t\prec g^{t+\tau},\,g_{t+\tau}=g}\beta^{t+\tau}\pi(g^{t+\tau}|g_0)\\
=&q^I\left[1+\sum_{\tau=1}^{\infty}\sum_{g^{t+\tau}\in\mathcal{G}^{t+\tau}:\,g^t\prec g^{t+\tau},\,g_{t+\tau}=g}\beta^\tau\frac{\pi(g^{t+\tau}|g_0)}{\pi(g^t|g_0)}\right]\\
=&q^I\left[1+\sum_{\tau=1}^{\infty}\sum_{g^{t+\tau}\in\mathcal{G}^{t+\tau}:\,g^t\prec g^{t+\tau},\,g_{t+\tau}=g}\beta^\tau\pi(g^{t+\tau}|g^t)\right]\\
\ge&q^I\left[1+\sum_{\tau=1}^{\infty}\beta^\tau \pi_{\min}\right]=q^I\left[1+\frac{\beta}{1-\beta}\pi_{\min}\right].
\end{aligned}
\end{equation}
We combine \eqref{eq:estimate_Greater} and \eqref{eq:estimate_Smaller} and obtain
$$
M_g>1+\frac{\beta}{1-\beta}\pi_{\min},
$$
or equivalently,
$$
\beta<\frac{M_g-1}{M_g-1+\pi_{\min}},
$$
which contradicts the assumption on $\beta$.

\textbf{Step 3.} We know that $\tilde{p}_k^I\le p_k$ for each $I\ge 0$, and $\sum_{k=1}^{M_g}\tilde{p}_{k}^I=\sum_{i=1}^{I}q^i\rightarrow
1=\sum_{k=1}^{M_g}p_k(I\rightarrow\infty)$. Therefore, $\tilde{p}_k^I\rightarrow p_k$ as $I\rightarrow\infty$ for each $k\in\{1,\cdots, M_g\}$.

\textbf{Step 4.} For any \(t\ge 1\) and \(g^t\in\mathcal{G}^t\) with \(g_t=g\), suppose that the term
\[
\frac{\beta^{t}\pi(g^{t}\mid g_0)}
{\sum_{s=1}^{\infty}\sum_{h^s\in\mathcal{G}^s:\,h_s=g}\beta^s\pi(h^s\mid g_0)}
\]
corresponds to \(q^i\) in the non-increasing rearrangement in \textbf{Step 1}. Then we set
$$
k_t(g^t)=k_i,
$$
where \(k_i\) is defined in \textbf{Step 2}. Therefore, we have
$$
\frac{1}
{\sum_{s=1}^{\infty}\sum_{h^s\in\mathcal{G}^s:\,h_s=g}\beta^s\pi(h^s\mid g_0)} \sum_{t=1}^{\infty}\sum_{g^t\in\mathcal{G}^t:\,g_t=g}\left[\beta^t\pi(g^t|g_0)\right]\vec{e}_{k_t(g^t)}=\lim_{I\rightarrow\infty}\vec{\tilde{p}}^I=\vec{p},
$$
and we finish the proof. \hfill $ \Box $
        \end{itemize}

\subsubsection*{Proof of Lemma \ref{lem:infinite_Ramsey_convex_twostate_text}}
    By straightforward computation, we have
    $$
    f'(\ell_t;b_{t-1},g_t)=-\frac{b_{t-1}}{(\ell_t-g_t)^2}+\frac{1}{2}\frac{1-\frac{\ell_t}{2}}{(1-\ell_t)^{\frac{3}{2}}},
    $$
    and
    $$
    f''(\ell_t;b_{t-1},g_t)=\frac{2b_{t-1}}{(\ell_t-g_t)^3}+\frac{1}{2}\frac{1-\frac{\ell_t}{4}}{(1-\ell_t)^{\frac{5}{2}}}.
    $$
    When $b_{t-1}\ge 0$, we know that $f''(\ell_t;b_{t-1},g_t)>0$, and hence $f$ is convex. Since $f$ is continuous, we conclude that $f$ has a unique minimal point $g_t+\epsilon\le \ell_t^*(b_{t-1},g_t)\le 1-\epsilon$. \newline 
    Furthermore, for any $0\le b_{t-1}^{1}<b_{t-1}^2$, we know that
    $$
    \begin{aligned}
        f^*(b_{t-1}^1,g_t)=&f(\ell_t^*(b_{t-1}^1,g_t);b_{t-1}^1,g_t)\\\le& f(\ell_t^*(b_{t-1}^2,g_t);b_{t-1}^1,g_t)\\\le& f(\ell_t^*(b_{t-1}^2,g_t);b_{t-1}^2,g_t)=f^*(b_{t-1}^2,g_t),
    \end{aligned}
    $$
    and
    $$
    f^*(b_{t-1}^1,g_t)=\frac{b_{t-1}^1}{\ell_t-g_t}+\frac{\ell_t}{2\sqrt{1-\ell_t}}-1\ge 0+\frac{g_t+\epsilon}{2\sqrt{1-(g_t+\epsilon)}}-1>-1.
    $$
    Therefore,  $f^*(b_{t-1},g_t)\ge \frac{g_t+\epsilon}{2\sqrt{1-(g_t+\epsilon)}}-1>-1$ is increasing w.r.t. $b_{t-1}$. \hfill $\Box$
\subsubsection*{Proof of Proposition \ref{prop:prop_bn1_twostate}}
    According to the definition, 
    $$
    f(\ell_0;b_{-1},g_0)=\frac{b_{-1}}{\ell_0}+\frac{\ell_0}{2\sqrt{1-\ell_0}}-1.
    $$
    Therefore, $f^*(b_{-1},g_0)=\frac{1}{k}$ if and only if
    $$
    \frac{b_{-1}}{\ell_0}+\frac{\ell_0}{2\sqrt{1-\ell_0}}-1\ge \frac{1}{k},\quad \forall \epsilon\le \ell_0\le 1-\epsilon,
    $$
    and
    $$
    \frac{b_{-1}}{\ell_0^*(b_{-1},g_0)}+\frac{\ell_0^*(b_{-1},g_0)}{2\sqrt{1-\ell_0^*(b_{-1},g_0)}}-1= \frac{1}{k},
    $$
    for some $\epsilon\le \ell_0^*(b_{-1},g_0)\le 1-\epsilon$; or equivalently,
    \begin{equation}\label{eq:lem_prop_bn1_property}
    \max_{\epsilon\le \ell_0\le 1-\epsilon}\ell_0\left( 1+\frac{1}{k}-\frac{\ell_0}{2\sqrt{1-\ell_0}}\right)=b_{-1}.
    \end{equation}
    The existence is straightforward from \eqref{eq:lem_prop_bn1_property} for any $K_1>0$. We take $$K_1=100,\,C_1=\frac{g^*}{4}\left(1-\frac{g^*}{4\sqrt{1-\frac{g^*}{2}}}\right).$$
    When $k>K_1$, $\epsilon=e^{-3k}<\frac{g^*}{2}<1-e^{-3k}=1-\epsilon$. On one hand, taking $\ell_0=\frac{g^*}{2}$, we have
    \begin{equation}\label{eq:lem_prop_bn1_mid1}
    \max_{\epsilon\le \ell_0\le 1-\epsilon}\ell_0\left( 1+\frac{1}{k}-\frac{\ell_0}{2\sqrt{1-\ell_0}}\right)\ge \frac{g^*}{2} \left(1-\frac{g^*}{4\sqrt{1-\frac{g^*}{2}}}\right)=2C_1.
    \end{equation}
    On the other hand, if we restrict $\ell_0$ to the interval $[\epsilon,C_1]$, we have
    \begin{equation}\label{eq:lem_prop_bn1_mid2}
    \max_{\epsilon\le \ell_0\le C_1}\ell_0\left( 1+\frac{1}{k}-\frac{\ell_0}{2\sqrt{1-\ell_0}}\right)\le (1+\frac{1}{100})C_1<2C_1. 
    \end{equation}
    The inequalities \eqref{eq:lem_prop_bn1_mid1} and \eqref{eq:lem_prop_bn1_mid2} imply that $\ell_0^*(b_{-1},g_0)>C_1$. \hfill $\Box$

\subsubsection*{Proof of Proposition \ref{prop:small_debt}}
 When $g_t=g^2$, according to \eqref{eq:ramsey_accumulate_infty_twostate} and the transversality condition, we have
$$
\begin{aligned}
\frac{b_{t-1}}{\ell_t-g^2}&\le \mathbb{E}_0\sum_{t=0}^{\infty}\beta^th(\ell_t)\le \frac{1}{1-\beta(1-\delta)}h(g^2+\epsilon)+\frac{\beta\delta}{(1-\beta)(1-\beta(1-\delta))}h(g^1+\epsilon)
\\&=\frac{1}{1-\beta(1-\delta)}\int_{g^*-\sigma+\epsilon}^{g^*}\left(-h'(\ell)\right) d \ell+\frac{\beta\delta}{(1-\beta)(1-\beta(1-\delta))}h(\epsilon)\\&=\frac{-h'(g^*)}{1-\beta}\sigma+o(\sigma),
\end{aligned}
$$
which implies the conclusion.\footnote{We in fact prove a stronger result that $\frac{b_t}{\ell_t-g^2}\le C_2\sigma$.}
\hfill $\Box$

\subsubsection*{Proof of Proposition \ref{prop_lb_twostate}}
We consider $k>\max\{K_1,K_2\}$ and proceed in two steps.
\begin{itemize}
\item \textbf{Step 1. Construct a feasible lottery solution.} 

We choose $\ell_0=\ell_0^*(b_{-1},g_0),\, b_0=\frac{\sigma}{2}h(g^*-\frac{\sigma}{2}):=C_3(\sigma)\sigma^2$, where $C_3(\sigma)\rightarrow -\frac{1}{4}h'(g^*)$ as $\sigma\rightarrow0$; and $b_t=0$, $\forall t\ge 1$. 

At $t=1$, when $g_t=g^2$, we take $\ell_1(g^2)=g^*-\frac{\sigma}{2}$.\footnote{This requires $\epsilon<\frac{\sigma}{2}$, which is always true when $k>1$.}  Then 
$$f(\ell_1(g^2);b_0,g^2)=\frac{\frac{\sigma}{2}h(g^*-\frac{\sigma}{2})}{(g^*-\frac{\sigma}{2})-(g^*-\sigma)}-h(g^*-\frac{\sigma}{2})=0,$$
and the forward looking constraint \eqref{eq:ramsey_accumulate_infty_twostate} with $t=1,\,g_t=g^2$ is satisfied.

At $t=1$, when $g_t=g^1$, the forward looking constraint \eqref{eq:ramsey_accumulate_infty_twostate} with $t=1,\,g_t=g^1$ is written
\begin{equation}\label{eq:g2g3cons_proplb}
f(\ell_1(g^1);b_{0},g^1)=\frac{C_3(\sigma)\sigma^2}{\ell_1(g^1)}-1+\frac{\ell_1(g^1)}{2\sqrt{1-\ell_1(g^1)}}\le 0.
\end{equation}
It is straightforward to verify that there exists $K_3>0$, s.t. when $k>K_3$, there exists $\ell^L_1(g^1):=2C_3(\sigma)\sigma^2<\ell^H_1(g^1):=\frac{g^*}{2}$ s.t. the inequality \eqref{eq:g2g3cons_proplb} holds for $\ell^L_1(g^1)$ and $\ell^H_1(g^1)$. We allow $\ell_1(g^1)$ to be a lottery that takes the value $\ell_1^L(g^1)$ with probability $p^L$ and $\ell_1^H(g^1)$ with probability $p^H$, then the constraint \eqref{eq:ramsey_accumulate_infty_twostate} for $t=0$ can be relaxed to
\begin{equation}\label{eq:cons0_lot}
\frac{1}{k}=f(\ell_0^*(b_{-1},g_0);b_{-1},g_0)\le \beta (1-\delta)\left[p^L\frac{b_{0}}{\ell^L_{1}(g_1)}+p^H\frac{b_{0}}{\ell^H_{1}(g_1)}\right]+\beta \delta \frac{b_0}{\left(g^*-\frac{\sigma}{2}\right)-\left(g^*-\sigma\right)}.
\end{equation}
When $k>\frac{2}{\beta(1-\delta)}$, we take $p^L=2/(\beta(1-\delta)k)$ and $p^H=1-2/(\beta(1-\delta)k)$. It is straightforward to verify that \eqref{eq:cons0_lot} is satisfied.

For $t\ge 2$, we choose $\ell_t\equiv g^*$. The constraint \eqref{eq:ramsey_accumulate_infty_twostate} for $t\ge 2$ is straightforward to verify.
\item \textbf{Step 2. Estimate the utility of this solution.} 

When $k>\max\{1,K_3,\frac{2}{\beta(1-\delta)}\}$,  the utility of this feasible lottery solution is
$$
\small
\begin{aligned}
&\underbrace{\log \ell_0+\sqrt{1-\ell_0}}_{t=0}\\&+\underbrace{\beta (1-\delta)\left[\frac{2}{\beta(1-\delta)k}(\log \ell_1^L(g^1)+\sqrt{1-\ell_1^L(g^1)})+(1-\frac{2}{\beta(1-\delta)k})(\log \ell_1^H(g^1)+\sqrt{1-\ell_1^H(g^1)})\right]}_{t=1,\,g_t=g^1}\\
&+\underbrace{\beta \delta \left[\log (\ell_1(g^2)-g^2)+\sqrt{1-\ell_1(g^2)}\right]}_{t=1,\, g_t=g^2}+\underbrace{R}_{t\ge 2}\\
\ge &\underbrace{\log C_1}_{t=0}+\underbrace{\frac{2}{k}\log 2C_3(\sigma)\sigma^2+\beta(1-\delta )(1-\frac{2}{\beta(1-\delta)k})\log\frac{g^*}{2}}_{t=1,g_t=g^1}+\underbrace{\beta\delta \log(\frac{\sigma}{2})}_{t=1,g_t=g^2}+\underbrace{R}_{t\ge 2}\\
\ge &\underbrace{\log C_1}_{t=0}+\underbrace{(\frac{2}{k}\log 2C_3(\sigma)-4+\log \frac{g^*}{2})}_{t=1,g_t=g^1}-\underbrace{\beta e^{-5k}(k+\log 2)}_{t=1,g_t=g^2}+\underbrace{R}_{t\ge 2}.
\end{aligned}
$$
There exist $C_4,K_4>0$, s.t. when $k>K_4$, we have
$$
\underbrace{\log C_1}_{t=0}+\underbrace{(\frac{2}{k}\log 2C_3(\sigma)-4+\log \frac{g^*}{2})}_{t=1,g_t=g^1}-\underbrace{\beta e^{-5k}(k+\log 2)}_{t=1,g_t=g^2}\ge C_4.
$$
It suffices to control $R$. Indeed, we have
$$
\begin{aligned}
R\ge& \sum_{t=2}^{\infty} \beta^tp(g_t=g^1)\log g^*+\sum_{t=2}^{\infty}\beta^t p(g_t\ne g^1)\log \epsilon\\
\ge &\frac{\beta^2}{1-\beta}\log g^* -3k\sum_{t=2}^{\infty}\beta^t\left(1-p(g_t=g^1)\right)\\
\ge &\frac{\beta^2}{1-\beta}\log g^*-3k\sum_{t=2}^{\infty}\beta^t\left(1-p(\{g_n=g^1,\,n=1,\cdots,t\})\right)\\
=&\frac{\beta^2}{1-\beta}\log g^*-3k\sum_{t=2}^{\infty}\beta^t\left[1-(1-\delta)^{t}\right]\\
\ge &\frac{\beta^2}{1-\beta}\log g^*-3ke^{-5k}\sum_{t=2}^{\infty}t\beta^t,
\end{aligned}
$$
where the second to last inequality holds by Bernoulli's inequality. There exist $C_5,K_5>0$, s.t. when $k>K_5$, we have
$$
\frac{\beta^2}{1-\beta}\log g^*-3ke^{-5k}\sum_{t=2}^{\infty}(t-1)\beta^t\ge C_5.
$$
Therefore, when $k>K_{\text{lot}}:=\max\{1,K_1,K_2,K_3,\frac{2}{\beta(1-\delta)},K_4,K_5\}$, the utility of this solution is not less than $C_{\text{lot}}:=C_4+C_5$. \hfill $\Box$
\end{itemize}

\subsubsection*{Proof of Proposition \ref{prop_ub_twostate}}
According to Proposition \ref{prop:small_debt}, when $k>K_2$, $b_0\le C_2\sigma$ and $\frac{b_0}{\ell_1(g^2)-g^2}\le  C_2\sigma$. 

Therefore, the constraint \eqref{eq:ramsey_accumulate_infty_twostate} for period 0 yields that
$$
\begin{aligned}
\frac{1}{k}=f^*(\ell_0(b_{-1},g_0);b_{-1},g_0)\le f(\ell_0;b_{-1},g_0)\le& \beta\mathbb{E}_0\frac{b_0}{\ell_1-g_1}\\=& \beta(1-\delta)\frac{b_0}{\ell_1(g^1)}+\beta \delta \frac{b_0}{\ell_1(g^2)-g^2}\\\le& \beta\frac{C_2\sigma}{\ell_1(g^1)}+\beta \delta C_2\sigma.
\end{aligned}
$$
Hence there exists $K_6>0$ such that, when $k>K_6$, 
$$
\frac{1}{2k}\le \beta\frac{C_2\sigma}{\ell_1(g^1)},
$$
implying that
$$
\ell_1(g^1)\le 2k\beta C_2\sigma.
$$
The utility at $t=1$ when $g_t=g^1$ is not greater than
 $$
\beta (1-\delta )\left[\log k+\log \sigma+\log2\beta C_2+\sqrt{1-0}\right]= \beta(1-\delta)\log \sigma+ \beta(1-\delta)\log k+ C(\beta,C_2).
$$
The total utility across periods is then bounded above by $\beta(1-\delta)\log \sigma+ \beta(1-\delta)\log k+\tilde{C}$ for some $\tilde{C}$ that is independent of $k$, due to the upper bound of $\log(\cdot-g_t)$ and $\sqrt{1-\cdot}$. \hfill $\Box$

\subsubsection*{Proof of Proposition \ref{prop:IM_lott_fb}}

We prove a stronger version in which we impose a debt restriction $b\in[\underline{b},\bar{b}]$, where $\underline{b}<0$ and $\bar{b}=\max\{b_0,1\}$. Clearly, $W_{\underline{c}}(b,0,g)$ is decreasing w.r.t. $b$. The proof can be divided into two steps.

\begin{itemize}
    \item \textbf{Step 1. Show that $W_{\underline{c}}(\underline{b},0,g)\le V^{\text{FB}}(g)\quad \forall g\in\mathcal{G}$.} According to the functional equation \eqref{eq:fe}, we have
    $$
    \begin{aligned}
    W_{\underline{c}}(\underline{b},0,g)\le& \sup_{\ell,b'}u(\ell-g)+v(\ell)+\beta\mathbb{E}_{g}W_{\underline{c}}(b',0,g')\\
    = &u(\ell^{\text{FB}}(g)-g)+v(\ell^{\text{FB}}(g))+\beta\mathbb{E}_gW_{\underline{c}}(\underline{b},0,g'),\quad \forall g\in\mathcal{G},
    \end{aligned}
    $$
    implying that $W_{\underline{c}}(\underline{b},0,g)\le V^{\text{FB}}(g),\quad \forall g\in\mathcal{G}.$
    \item \textbf{Step 2. Show that $\liminf_{\underline{c}\rightarrow\zeta}W_{\underline{c}}(\bar{b},0,g)\ge V^{\text{FB}}(g)\quad \forall g\in\mathcal{G}$.} 
    According to the functional equation \eqref{eq:fe}, taking $\ell_0=\ell^{\text{FB}}(g),\, b_0=\frac{\underline{c}}{2}$, we have
    \begin{equation}\label{eq:estimatemain_ICM}
    \begin{aligned}
    W_{\underline{c}}(\bar{b},0,g)\ge& \inf_{\gamma\ge 0}D_g(\gamma;\underline{c}):=u(\ell^{\text{FB}}(g)-g)+v(\ell^{\text{FB}}(g))\\&+\gamma\left[u'(\ell^{\text{FB}}(g)-g)(\ell^{\text{FB}}(g)-g)+v'(\ell^{\text{FB}}(g))\ell^{\text{FB}}(g)-\bar{b}u'(\ell^{\text{FB}}(g)-g)\right]\\
    &+\beta\mathbb{E}_{g}W_{\underline{c}}(\frac{\underline{c}}{2},\gamma,g').
    \end{aligned}
    \end{equation}
    When $\underline{c}-\zeta \rightarrow 0$, for the problem with initial debt $\frac{\underline{c}}{2}$ and $g'\in\mathcal{G}$, one can consider a feasible allocation $c_t\equiv \underline{c},\,b_t\equiv \frac{\underline{c}}{2}$, and have
    \begin{equation}\label{eq:estimate1_ICM}
    W_{\underline{c}}(\frac{\underline{c}}{2},\gamma,g')\ge \gamma \frac{\underline{c}}{2}u'(\underline{c})+\frac{1}{1-\beta}\left[u(\underline{c})+v(\bar{\ell})\right].
    \end{equation}
    We combine \eqref{eq:estimatemain_ICM} and \eqref{eq:estimate1_ICM}, and obtain
    $$
    D_g(\gamma;\underline{c})\ge \frac{1}{2}\underline{c}u'(\underline{c})\left[\gamma(\beta+I_g(\underline{c}))+J_g(\underline{c})\right],
    $$
    and
    $$
    \inf_{\gamma\ge 0}D_g(\gamma;\underline{c})\le W_{\underline{c}}(\bar{b},0,g)\le V^{\text{FB}}(g),
    $$
    where
    $$
    \begin{aligned}
    I_g(\underline{c})&:=\frac{2\left[u'(\ell^{\text{FB}}(g)-g)(\ell^{\text{FB}}(g)-g)+v'(\ell^{\text{FB}}(g))\ell^{\text{FB}}(g)-\bar{b}u'(\ell^{\text{FB}}(g)-g)\right]}{u'(\underline{c})\underline{c}}\rightarrow0,\\J_g(\underline{c})&:=2\frac{u(\ell^{\text{FB}}(g)-g)+v(\ell^{\text{FB}}(g))}{u'(\underline{c})\underline{c}}+\frac{2\beta}{1-\beta}\left[\frac{u(\underline{c})}{u'(\underline{c})\underline{c}}+\frac{v(\bar{\ell})}{u'(\underline{c})\underline{c}}\right]\rightarrow0,\,\underline{c}u'(\underline{c})\rightarrow\infty,\quad \text{ as }\underline{c}\rightarrow\zeta.
    \end{aligned}
    $$
    Therefore, for any $\gamma>0$, $D_g(\gamma;\underline{c})\rightarrow\infty(\underline{c}\rightarrow\zeta)$. Hence there exists $\bar{\gamma}_g(\underline{c})>0$ satisfying $\bar{\gamma}_g(\underline{c})\rightarrow0$, s.t. 
    $$
    W_{\underline{c}}(\bar{b},0,g)\ge \inf_{\gamma\ge 0}D_g(\gamma;\underline{c})=\inf_{0\le \gamma \le \bar{\gamma}_g(\underline{c})}D_g(\gamma;\underline{c}).
    $$
    Since $W_{\underline{c}}(\frac{\underline{c}}{2},\gamma,g')\ge W_{\underline{c}}(\frac{\underline{c}}{2},0,g')\ge W_{\underline{c}}(\bar{b},0,g')$, according to \eqref{eq:estimatemain_ICM} we have
    $$
    \begin{aligned}
    W_{\underline{c}}(\bar{b},0,g)\ge& \inf_{0\le \gamma \le \bar{\gamma}_g(\underline{c})}D_g(\gamma;\underline{c})\\
    \ge& u(\ell^{\text{FB}}(g)-g)+v(\ell^{\text{FB}}(g))+\beta\mathbb{E}_{g}W_{\underline{c}}(\bar{b},0,g')\\&\underbrace{-\bar{\gamma}_g(\underline{c})\left|u'(\ell^{\text{FB}}(g)-g)(\ell^{\text{FB}}(g)-g)+v'(\ell^{\text{FB}}(g))\ell^{\text{FB}}(g)-\bar{b}u'(\ell^{\text{FB}}(g)-g)\right|}_{I},
    \end{aligned}
    $$
    where $I\rightarrow0$ as $\underline{c}\rightarrow \zeta$, implying that
    $$\liminf_{\underline{c}\rightarrow\zeta}W_{\underline{c}}(\bar{b},0,g)\ge V^{\text{FB}}(g),\quad \forall g\in\mathcal{G}.
    $$
 \end{itemize}

\section{The computational method} \label{app:algo}
As explained in \cite{shen2025recursive}, a computational method that iterates on
\begin{equation}
\label{eq:iter}
  W^{n+1}(b,\mu,g)=\sup_{\ell,\,\bp}\ \inf_{\gamma\ge0}\
  \Big\{u(c)+v(\ell)+\mu b u'(c)+\gamma\big[u'(c)(c-b)+\ell v'(\ell)\big]
        +\beta\,\mathbb{E}_g\,W^n(\bp,\gamma,g')\Big\}.
\end{equation}
converges to some fixed point of the functional equation (FE) \eqref{eq:fe}, and for a sufficiently large\footnote{In practice, the starting point turns out not to matter for the computation. In theory, one should take a sufficiently large function that is unbounded as $ \mu \rightarrow \infty$.} starting point $ W^0 $, this fixed point gives the optimal value of the planner's problem with ex-post lotteries \eqref{eq:Ramsey-l}.
We discretize savings and labor supply, but leave the multiplier to be continuous.
Suppose the set of possible labor choices $\hat{\mathcal{L}}=\{\ell_i\}_{i=1}^{n_{\ell}}$ contains $ n_{\ell} $ elements, assumed to lie strictly between zero and one, and the set of possible bond choices $\hat{\mathcal{B}}=\{b_j\}_{j=1}^{n_b}$ contains $ n_b $ elements. By feasibility, there is also a set of possible consumption choices $\hat{\mathcal{C}}=\{c_i\}_{i=1}^{n_{\ell}}$ containing $ n_{\ell} $ elements.
For discretized labor and bond choices, the solution to the FE can be computed exactly (up to machine precision).

\subsection{Theoretical foundations of the method}
To explain the algorithm we first state the  following two lemmas.
\begin{lemma}\label{lem:env}
Fix $(b,g)\in\hat{\mathcal
{B}}\times \mathcal{G}$, and suppose that for every $\bp\in\hat{\mathcal{B}}$ and every $g'\in\mathcal{G}$ the
continuation $W^n(\bp,\cdot,g')$ is continuous, and $ W^{n+1} $ solves (\ref{eq:iter}). Then either $W^{n+1}(b,\cdot,g)\equiv-\infty$, or
$W^{n+1}(b,\cdot,g)$ is
finite, convex, piecewise linear and continuous in $\mathbb{R}_+$, with
at most $n_{\ell}$ pieces.
\end{lemma}

\begin{proof}
For a given $ (b,g)\in\hat{\mathcal{B}}\times\mathcal{G} $ and fixed labor choice $\ell_i\in\hat{\mathcal{L}}$ and bond choice $\bp_j\in\hat{\mathcal{B}}$, we write the expression inside
\eqref{eq:iter} as
\[
  \underbrace{u(c_i)+v(l_i)}_{\text{free of }\mu,\gamma}
  \;+\;\mu\,b\,u'(c_i)\;+\;\underbrace{\gamma\underbrace{\big[u'(c_i)(c_i-b)+\ell_i v'(\ell_i)\big]}_{s_i(b,g)}
        +\;\beta\,\underbrace{\mathbb{E}_g\,W^n(\bp_j,\gamma,g')}_{\bar{W}^n(b_j',\gamma,g)}}_{\psi_{i,j}(\gamma)}.
\]
Note that
$$ \max_{j} \inf_{\gamma \ge 0}  \left\{ u(c_i)+v(\ell_i) + \mu\,b\,u'(c_i) +\psi_{i,j}(\gamma) \right\} = 
u(c_i)+v(\ell_i) + \mu\,b\,u'(c_i) + \max_{j} \inf_{\gamma \ge 0} \psi_{i,j}(\gamma),
$$
and therefore, for each $i$, the expression above is an affine function in $ \mu $ if it is finite.
Let $m_i=\max_j \inf_{\gamma\ge0}\psi_{i,j}(\gamma)$, and note that
\[
W^{n+1}(b,\mu,g)=\max_i\left\{u(c_i)+v(\ell_i)+m_i+\mu b\,u'(c_i)\right\}
\]
is the pointwise maximum of $n_\ell$ affine functions (or functions identically equal to $-\infty$ when $m_i=-\infty$). Therefore, it is either identically $-\infty$ or a convex, piecewise linear, and continuous function.
\end{proof}

Together with the monotone convergence result from \cite{shen2025recursive}, the lemma implies that there exists a convex, continuous, and piecewise linear solution to the FE \eqref{eq:fe} whose value is the optimal value of \eqref{eq:Ramsey-l}. 

The lemma also directly suggests a computational method for solving the sup-inf problem \eqref{eq:iter} given $W^n$. Starting from a function $W^0$ that is convex, piecewise linear, and continuous on $\mathbb{R}_+$, the iteration preserves these properties for all $W^n$. Given $W^n$, it follows that $\psi_{i,j}(\cdot)$ is also convex and piecewise linear by definition. The following lemma then shows that each subproblem $\inf_{\gamma}\psi_{i,j}(\gamma)$ can be solved efficiently; the sup-inf problem can then be solved by combining these subproblems with an exhaustive search over $(i,j)$.
\begin{lemma}\label{lem:inner}
Let $\psi_{i,j}(\cdot)$ be convex piecewise linear with at most $K$ pieces identified by points $X_1,\cdots, X_K$.  For each $ k=1,\ldots,K $, denote its slope on
$[X_k,X_{k+1})$ by $ S_k$. Then
\begin{enumerate}
\item $\inf_{\gamma\ge0}\psi_{i,j}=-\infty$ if and only if $S_{K}<0$;
\item otherwise the infimum is attained at $\gamma^{*}=X_{k^{*}}$, where $k^{*}$ is
the first index with $ S_{k^{*}}\ge0$.
\end{enumerate}
\end{lemma}

\begin{proof}
$\psi_{i,j}$ is convex and piecewise linear with right-derivative $S_k$ on the
$k$-th segment. On the last, unbounded segment $\psi_{i,j}$ decreases without bound
exactly when that slope is negative. Otherwise a convex function on $[0,\infty)$ is
minimized where its right-derivative first becomes non-negative, which is the left
end of segment $k^{*}$.
\end{proof}
\subsection{Computing the solution of the FE}
Starting from a large function $W^0$, each iteration records the intercepts and slopes 
of $ u(c_i)+v(\ell _i) + m_i + \mu b u'(c_i) $ given $(b,g)\in\hat{\mathcal{B}}\times\mathcal{G}$. The slopes $a_i(b,g)=b\,u'(c_i)$ are computed once\footnote{ In the computations below we take $ g $ to take two values, and we use the
resolution $n_{\ell}=153$, $n_b=305$, resulting in 93330 potential slopes}. 
In each iteration of  (\ref{eq:iter}) it suffices to keep track of the array of intercepts, $F_i^n(b,g)=u(c_i)+v(\ell_i)+m_i$,  which consists of $n_{\ell}\times n_b\times2$ numbers,
with $-\infty$ allowed.  
Each iteration then maps $F$ to a new array $F^{+}$ in four steps:

\emph{Step 1: the envelope of each continuation on $[0,\infty)$.} For every
$(\bp_j,g')\in\hat{\mathcal{B}}\times\mathcal{G}$, sort the finite lines by slope $a_i(b'_j,g')$ (keeping the larger intercept $F_i(b'_j,g')$ at equal
slopes) and build the upper envelope with the standard monotone stack: a new line
removes the line on top of the stack whenever it overtakes it no later than that
line took over. Clip the result to $\gamma\ge0$. The output is a list of
breakpoints $0=x_1<x_2<\dots<x_K$ and the active line on each segment. The cost is
$O(n_{\ell}\log n_{\ell})$.

\emph{Step 2: the expected continuation.} $ \bar W(\bp_j,\cdot,g)$ is a
$\pi_{gg'}$-weighted sum of the two envelopes. Merging the two sorted breakpoint lists gives
\[
  \bar W(\bp_j,\gamma,g)=C_k+\gamma S_k\qquad\text{for }\gamma\in[X_k,X_{k+1}),
\]
with $X_1=0$, the breakpoints $X_k$ the union of the two envelopes' breakpoints,
and $S_k$, $C_k$ the probability-weighted slope and intercept of the active lines.
Each merged breakpoint records which envelope and which pair of labor lines it
came from. If $\pi_{gg'}>0$ and $W(\bp_j,\cdot,g')\equiv-\infty$, then
$\bar W(\bp_j,\cdot,g)\equiv-\infty$.

\emph{Step 3: the inner infimum, exactly.} $\psi_{i,j}(\gamma)$ can be computed by $\psi_{i,j}(\gamma)=\gamma s_i(b,g)+\beta\bar{W}(b_j',\gamma,g)$. 
We apply Lemma \ref{lem:inner}.
Because the slopes $S_k$ are sorted, $k^{*}$ is found by binary search in
$O(\log K)$ steps, with $K\le2n_{\ell}$. 

\emph{Step 4: update.} For each $(b,g)\in\mathcal{\hat{B}}\times\mathcal{G}$ and feasible labor point $\ell_i\in\hat{\mathcal{L}}$,
\[
\begin{aligned}
  F^{+}_i(b,g)=&u(c_i)+v(\ell_i)+m_i\\
  =&u(c_i)+v(\ell_i)+\max_{j}\ \inf_{\gamma\ge0}\big\{\gamma s_i(b,g)+\beta\bar W(\bp_j,\gamma,g)\big\},
\end{aligned}
\]
which is $-\infty$ if the infimum is unbounded for every $\bp_j\in\hat{\mathcal{B}}$. The maximizer $b^{'*}_j$,
the multiplier $\gamma^{*}$ and the index $k^{*}$ are stored as the policy of line
$i$.

\subsection{Optimal Policies}
As explained in detail in \cite{shen2025recursive},  to recover policies, one needs to carry, besides debt, the multiplier and  the
spending state, a \emph{promise} (the continuation of the Euler equations) $\phi$ along the path. In the Ramsey setting, this is simply the value of $b\,u'(c)$ that last period's constraint requires, that is, the promised marginal utility of consumption times debt.
By Lemma~\ref{lem:env} the value function is the upper envelope of the lines
$F_i(b,g)+\mu\,b\,u'(c_i)$, so the slope of a line is exactly the promise its allocation
delivers, and $\partial W(b,\cdot,g)$ is the set of promises the planner can deliver. Three
facts make the recovery immediate:
\begin{itemize}
\item a promise strictly between the slopes of two lines that are tied at a kink $\mu^{*}$
can only be delivered by randomizing over them, with weights
$w=(\phi-b\,u'(c_1))/(b\,u'(c_2)-b\,u'(c_1))$;
\item each branch keeps its own continuation $(\bp,\gamma^{*})$ and, hence, satisfies its own
implementability constraint. 
\item Every interior multiplier lands on a kink of next
period's value function, so the promise handed forward is generally set-valued. Which
element is passed on is pinned by complementary slackness: a branch with $\gamma^{*}>0$ must
have its own constraint hold with equality.
\end{itemize}
The value of the resulting lottery is $\sum_i w_iF_i(b,g)$, which equals
$\min_\mu\{W(b,\mu,g)-\mu\phi\}$; in the computed examples the two agree to $10^{-6}$ or better, and
the promise is delivered to machine precision.

\section{Detailed derivations} \label{app:sec2}
\subsection{The complete markets problem} \label{app:cm_derivation}
Consider the problem
\begin{align}  \label{eq:ramsey_obj1}&\max_{(c_t,\ell_t,\tau_t,T_t\ge 0)_{t=0}^{\infty},(q_t,b_t)_{t=1}^{\infty}} \mathbb{E}_{g_0}\sum_{t=0}^{\infty}\beta^t\left[u(c_t)+v(\ell_t)\right] , \\ 
\label{eq:ramsey_ponzi}\mbox{ s.t. } &\sup_{t\ge 0,\, g^t\in\mathcal{G}^t}|b_t(g^t)|<\infty;\\&\forall t\ge 0,\,g^t\in\mathcal{G}^t:\notag\\
&\label{eq:ramsey_agentrc0}c_t(g^t)=b_t(g^t)+(1-\tau_t(g^t))\ell_t(g^t)-\sum_{g_{t+1}\in\mathcal{G}}q_{t+1}(g^{t+1})b_{t+1}(g^{t+1})+T_t(g^t);\\ 
& (1-\tau_t(g^t)) u'(c_t(g^t))=-v'(\ell_t(g^t)); \\ \label{eq:ramsey_agenteuler}
& -q_{t+1}(g^{t+1}) u'(c_t(g^t))+\beta \pi(g_{t+1}|g_{t})u'(c_{t+1}(g^{t+1})) = 0,\quad \forall g_{t+1}\in\mathcal{G}; \\ 
& \sum_{g_{t+1}\in\mathcal{G}}q_{t+1}(g^{t+1}) b_{t+1}(g^{t+1}) = b_{t}(g^t) - \tau_t(g^t) \ell_t(g^t) + g_t+T_t(g^t).
\label{eq:ramsay_gbc}
\end{align}

The following lemma gives a reduced formula for the Ramsey problem.
\begin{lemma}\label{lem:ramsey_conssimplified}
    Given $\{(\ell_t)_{t=0}^{\infty},b_{0}\}$, the following are equivalent:
    \begin{enumerate}
        \item  There exists $(c_t,\tau_t,T_t\ge 0)_{t=0}^{\infty},(q_t,b_t)_{t=1}^{\infty}$, s.t. $\{(c_t,\ell_t,\tau_t,T_t\ge 0)_{t=0}^{\infty},(q_t,b_t)_{t=1}^{\infty}\}$ satisfies \eqref{eq:ramsey_ponzi} to \eqref{eq:ramsay_gbc};
        \item The inequality
\begin{equation}\label{eq:ramsey_accumulte0}
    b_0u'(c_0)\le \mathbb{E}_{g_0}\sum_{t=0}^{\infty}\beta^t\left[u'(c_t)c_t+v'(\ell_t)\ell_t\right],
\end{equation}
where $c_t(g^t)=\ell_t(g^t)-g_t,\,\forall t\ge 0,\,g^t\in\mathcal{G}^t$
is satisfied.
    \end{enumerate}
\end{lemma}
\begin{proof}
We first show that $1\Rightarrow 2$. Since
$c_t(g^t)=\ell_t(g^t)-g_t,\,\forall t\ge 0,\,g^t\in\mathcal{G}^t$,
we have
    \begin{equation}\label{eq:lem_reduce_budget}
    \begin{aligned}
  b_t(g^t)=&c_t(g^t)-(1-\tau_t(g^t))\ell_t(g^t)+\sum_{g_{t+1}\in\mathcal{G}}q_{t+1}(g^{t+1})b_{t+1}(g^{t+1})-T_t(g^t)\\\le& c_t(g^t)-(1-\tau_t(g^t))\ell_t(g^t)+\sum_{g_{t+1}\in\mathcal{G}}q_{t+1}(g^{t+1})b_{t+1}(g^{t+1})\\=&c_t(g^t)+\frac{v'(\ell_t(g^t))}{u'(c_t(g^t))}\ell_t(g^t)+\beta\frac{\sum_{g_{t+1}\in\mathcal{G}}\pi(g_{t+1}|g_t)u'(c_{t+1}(g^{t+1}))b_{t+1}(g^{t+1})}{u'(c_t(g^t))},
\end{aligned}
\end{equation}
We multiply both sides of \eqref{eq:lem_reduce_budget} by $u'(c_t(g^t))$ and obtain
\begin{equation}\label{eq:lem_reduce_recursive_budget}
b_t(g^t)u'(c_t(g^t))\le  \left[u'(c_t(g^t))c_t(g^t)+v'(\ell_t(g^t))\ell_t(g^t)\right]+\beta\sum_{g_{t+1}\in\mathcal{G}}\pi(g_{t+1}|g_t)u'(c_{t+1}(g^{t+1}))b_{t+1}(g^{t+1}).
\end{equation}
The inequalities \eqref{eq:lem_reduce_recursive_budget} then give \eqref{eq:ramsey_accumulte0}.

Now we show that $2\Rightarrow 1$. If 
\eqref{eq:ramsey_accumulte0} is satisfied, for any $t\ge 0,\,g^t\in\mathcal{G}^t$ we construct
$$
\small
\begin{aligned}
c_t(g^t)&=\ell_t(g^t)-g_t,\, \tau_t(g^t)=1+\frac{v'(\ell_t(g^t))}{u'(c_t(g^t))}< 1,\,\\T_0&=\frac{1}{u'(c_0)}\mathbb{E}_{g_0}\sum_{t=0}^{\infty}\beta^t\left[u'(c_t)c_t+v'(\ell_t)\ell_t\right]-b_0,\\\forall g_{t+1}\in\mathcal{G}:\\
q_{t+1}(g^{t+1})&=\frac{\beta \pi(g_{t+1}|g_t) u'(c_{t+1}(g^{t+1}))}{u'(c_t(g^t))}>0,\\b_{t+1}(g^{t+1})&=\frac{1}{u'(c_{t+1}(g^{t+1}))}\left\{\mathbb{E}_{g^{t+1}}\sum_{n=0}^{\infty}\beta^n\left[u'(c_{t+1+n}(g^{t+1+n})c_{t+1+n}(g^{t+1+n})+v'(\ell_{t+1+n}(g^{t+1+n})\ell_{t+1+n}(g^{t+1+n})\right]\right\},\\T_{t+1}(g^{t+1})&=0.
\end{aligned}
$$
and it is easy to see that $\{(c_t,\ell_t,\tau_t,T_t)_{t=0}^{\infty},(q_t,b_t)_{t=1}^{\infty}\}$ satisfies \eqref{eq:ramsey_ponzi} to \eqref{eq:ramsay_gbc}.\end{proof}

\subsection{Regularity Assumptions}\label{app:regu}
In this subsection, we introduce a technical assumption on the problem that allows us to replace the sup-inf-sup formula by a max-min-max formula, thereby guaranteeing the attainability of \eqref{eq:expost_CM}. 

We equip the space $L$ with the product topology. Then $L$ is a compact set by Tychonoff's theorem. It is straightforward to verify that the maps
$$
\bar{F}:L\rightarrow \mathbb{R},\quad (\ell_t)_{t=1}^{\infty}\mapsto \mathbb{E}_{g_0}\sum_{t=1}^{\infty}\beta^tF(c_t,\ell_t)
$$
and
$$
\bar{G}:L\rightarrow \mathbb{R},\quad (\ell_t)_{t=1}^{\infty}\mapsto \mathbb{E}_{g_0}\sum_{t=1}^{\infty}\beta^tG(c_t,\ell_t)
$$
are continuous. Therefore, the third operator $\sup_{(\ell_t)_{t\ge 1}}$ can be replaced by $\max_{(\ell_t)_{t\ge 1}}$ without any additional assumptions. For the first two operators, we introduce the following assumption on the problem.
\begin{assumption}\label{ass:regu}
    The initial $b_0$ satisfies the condition that there exists a compact set $L_0\subseteq [\underline{c}+g_0,\bar{\ell}]$, s.t. \begin{itemize}
        \item the supremum of $\ell_0$ lies on $L_0$, i.e. $$\sup_{\ell_0\in L_0}\inf_{\lambda\ge 0}\sup_{(\ell_t)_{t=1}^\infty} \mathcal{L}((\ell_t)_{t=0}^{\infty};\lambda)=\sup_{\ell_0\in[\underline{c}+g_0,\bar{\ell}]}\inf_{\lambda\ge 0}\sup_{(\ell_t)_{t=1}^\infty} \mathcal{L}((\ell_t)_{t=0}^{\infty};\lambda)$$
        \item The uniform Slater's condition holds for $\ell_0\in L_0$, i.e. there exists $\epsilon>0$, s.t. 
        $$
        \max_{(\ell_t)_{t=1}^{\infty}\in L}\mathbb{E}_{g_0}\sum_{t=0}^{\infty}\beta^tG(c_t,\ell_t)-b_0u'(c_0)\ge \epsilon,\quad\forall \ell_0\in L_0.
        $$
    \end{itemize}
\end{assumption}
The uniform Slater's condition ensures that the second operator $\inf_{\lambda\ge 0}$ is attained at some $\lambda\in[0,\bar{\lambda}]$ uniformly for all $\ell_0\in L_0$, and hence it can be replaced by $\min_{0\le \lambda \le \bar{\lambda}}$. The compactness of $L_0$ and the continuity of
\[
\min_{0\le \lambda\le \bar{\lambda}}\max_{(\ell_t)_{t\ge 1}}\mathcal{L}(\ell_0,(\ell_t)_{t\ge 1};\lambda)
\]
with respect to $\ell_0$ ensure that the first operator $\sup_{\ell_0}$ can be replaced by $\max_{\ell_0\in L_0}$.

Assumption \ref{ass:regu} is not restrictive in many problems of interest. It only rules out some extreme cases in which, at the optimal labor supply in the first period, the only feasible subsequent labor supply is the one that satisfies the constraint $$
  b_0 u'(\ell_0-g_0) \le E_0 \sum_{t=0}^\infty \beta ^t (u'(\ell_t-g_t) (\ell_t-g_t) + v'(\ell_t) \ell_t)
$$ with equality. If one considers the discretized-$\ell_0$ problem and Slater's condition holds for some $\ell_0$, then the assumption is satisfied almost surely.

\subsection{Overview of numerical examples} \label{app:examples}
In order to establish \lq\lq often\rq\rq \ lotteries arise in calibrated examples, we solve 85 examples numerically. The parameter specifications are listed in Table \ref{tab:sec6-specs}.
\begin{table}[htbp]
\centering
\small
\caption{List of calibrations. Throughout, $u(c)=\log c$, $\beta=0.9$, $n_l=153$.
Spending is a two-state Markov chain with $\pi_{gg}=p$, $\pi_{gg'}=1-p$ ($p=0.5$ is i.i.d.).}
\label{tab:sec6-specs}
\begin{tabular}{@{}clcccll@{}}
\toprule
\# & $v(l)$ & $\eta$ & $\psi$ & $g$ & $p$ & $b_{\max}$ \\
\midrule
1  & $\psi(1-l)^{\eta}$ & $0.2,\ 0.5$        & 1   & $\{0,\,0.5\}$   & 0.5 & 0.5 \\
   & $\psi\log(1-l)$    & --                  & 1   & $\{0,\,0.5\}$   & 0.5 & 0.5 \\
\addlinespace
2  & $\psi\log(1-l)$    & --                  & 0.7 & $\{0.1,\,0.3\}$ & 0.5 & $1.5,\ 0.5,\ 0.25,\ 0.2,\ 0.15,\ 0.1$ \\
\addlinespace
3  & $\psi(1-l)$        & 1                   & --  & $\{0.1,\,0.3\}$ & 0.5 & $2.5,\ 4$ \\
\addlinespace
4  & $\psi(1-l)^{\eta}$ & $0.1,\dots,0.9$     & 1   & $\{0,\,0.5\}$   & 0.5 & $0.5,\ 1$ \\
\addlinespace
5  & $\psi(1-l)^{\eta}$ & $0.1,\dots,0.9$     & 1   & $\{0.1,\,0.3\}$ & 0.5 & $0.5,\ 1$ \\
\addlinespace
6  & $\psi(1-l)^{\eta}$ & 0.5                 & 1   & $\{0,\,0.5\}$   & $0.5,0.6,0.7,0.8,0.9,$ & 0.5 \\
   &                    &                     &     &                 & $0.92,0.94,0.95,0.97,0.99$ & \\
\addlinespace
7  & $\psi(1-l)^{\eta}$ & $0.1,\dots,0.5$     & 1   & $\{0.1,\,0.3\}$ & 0.97 & 0.5 \\
   & $\psi\log(1-l)$    & --                  & 1   & $\{0.1,\,0.3\}$ & 0.97 & 0.5 \\
   & $\psi\log(1-l)$    & --                  & 1   & $\{0,\,0.5\}$   & 0.97 & 0.5 \\
\addlinespace
8  & $\psi(1-l)^{\eta}$ & $0.1,\dots,0.5$     & 1   & $\{0.1,\,0.3\}$ & 0.99 & 0.5 \\
   & $\psi\log(1-l)$    & --                  & 1   & $\{0.1,\,0.3\}$ & 0.99 & 0.5 \\
\bottomrule
\end{tabular}
 
\end{table}

Of the $85$ calibrations solved, all exhibit genuine ties at every grid. At
$n_b=305$ and in the window $\mu\le2.45$ the count runs from $12$ (log leisure,
$g\in\{0,0.5\}$, $p=0.97$) to $299$. Where a calibration was solved on several grids, the count
grows roughly in proportion to $n_b$, as it must for a curve of ties in $(b,\mu)$.
Multiplicity of the value function is therefore the rule. For a typical calibration, there are always regions in the state space where lotteries become optimal.

Whether the multiplication is actually reached along a simulated path (and therefore whether lotteries are part of the optimal solution) is a much more subtle issue.
In Table \ref{table:hits} we list some simple cases where multiplicity turns out to be reachable. We list the number of \lq\lq hits\rq\rq, that is the number of ties that are reached with positive probability for some initial condition with $ \mu=0$.

\begin{table}[H]
\centering
\small
\begin{tabular}{lrrrr}
\toprule
& \multicolumn{3}{c}{hits, $n_l=153$} & hits\\
\cmidrule(lr){2-4}
calibration & $n_b\!=\!153$ & $305$ & $609$ & $n_l\!=\!305$\\
\midrule
A\ \ $\eta=0.2$, i.i.d., $g\in\{0,0.5\}$ & $17$ & $26$ & $36$ & $36$\\
B1\ $\eta=0.5$, $p=0.94$, $g\in\{0,0.5\}$ &  $3$ & $2$ & $12$ & $3$\\
B2\ $\eta=0.5$, $p=0.95$ & $4$ & $6$ & $16$ & $6$\\
B3\ $\eta=0.5$, $p=0.97$ & $5$ & $7$ & $8$ & $4$\\
B4\ $\eta=0.5$, $p=0.99$ &  $2$ & $3$ & $2$ & $6$\\
C1\ log leisure, $p=0.97$, $g\in\{0.1,0.3\}$ &  $0$ & $12$ & $36$ & $10$\\
C2\ log leisure, $p=0.99$, $g\in\{0.1,0.3\}$ & $10$ & $8$ & $14$ & $8$\\
D\ \ $\eta=0.5$, i.i.d., $g\in\{0,0.8\}$ & $11$ & $22$ & $26$ & $26$\\
\bottomrule
\end{tabular}
\caption{Lotteries along simulated paths, in the rightmost column $ n_b=305 $.}
\label{table:hits}
\end{table}

Hit counts are grid-sensitive because
refinement resolves more reachable co-states that land on the same ties. The
locations of those ties, and the multipliers that reach them, are not.

While in the listed cases, lotteries are part of the optimal solution,
many cases listed in Table \ref{tab:sec6-specs} exhibit zero hits. In these cases, lotteries are not part of the optimal solution. Table~\ref{table:nohits} lists them.

\begin{table}[H]
\centering
\footnotesize
\setlength{\tabcolsep}{4pt}
\begin{tabular}{llrl}
\toprule
family & cells & ties at $n_b\!=\!305$ & reachable multiplicity\\
\midrule
$g\in\{0,0.5\}$ i.i.d., $\eta=0.1$--$0.9$, $b_{\max}\in\{0.5,1\}$ & $18$ & $74$--$242$ & only case A\\
$g\in\{0.1,0.3\}$ i.i.d., $\eta=0.1$--$0.9$, $b_{\max}\in\{0.5,1\}$ & $18$ & $107$--$288$ & none\\
$g\in\{0,0.5\}$, $\eta=0.5$, $p=0.6$--$0.92$ & $5$ & $142$--$184$ & none\\
$g\in\{0,0.5\}$, $p=0.97$, $\eta=0.1$--$0.4$ & $4$ & $203$--$288$ & none\\
$g\in\{0.1,0.3\}$, $p\in\{0.97,0.99\}$, $\eta=0.1$--$0.5$ & $10$ & $249$--$299$ & none\\
$g\in\{0.1,0.3\}$, log, $\psi=0.7$, $b_{\max}=0.1$--$1.5$ & $6$ & $41$--$213$ & none\\
$g\in\{0.1,0.3\}$ i.i.d., $b_{\max}=1.5$, $\eta=0.1$--$1$ & $8$ & $62$--$216$ & none\\
$g\in\{0,g_2\}$ i.i.d., $\eta=0.5$, $g_2\in\{0.1,\ldots,0.4,0.6,0.7\}$ & $6$ & $50$--$285$ & none\\
$g\in\{0,0.5\}$, log leisure, i.i.d.\ and $p=0.97$ & $2$ & $0$ and $12$ & none\\
$g\in\{0.1,0.3\}$ i.i.d., $\eta=1$, $b_{\max}=4$ & $1$ & $25$ & none\\
\bottomrule
\end{tabular}
\caption{Calibrations without reachable lotteries}\label{table:nohits}
\end{table}

\end{appendices}

\end{document}